\documentclass[lettersize,journal]{IEEEtran}

\usepackage[T1]{fontenc}
\usepackage[utf8]{inputenc}
\usepackage[ngerman,english]{babel}
\usepackage{amsmath, amssymb, ﻿amsfonts}   
\usepackage{cuted} 
\usepackage{graphicx}         
\usepackage{algorithmic}
\usepackage{algorithm}
\usepackage{stfloats}
\usepackage{url}
\usepackage{verbatim}
\usepackage{graphicx}
\usepackage{cite}

\usepackage{upgreek}			  
\usepackage{mathtools}		  
\usepackage{url}					
\usepackage{enumitem}		
\usepackage{accents}			
\usepackage{bbm}
\usepackage{xcolor}
\usepackage{lipsum}
\usepackage{dirtytalk}
\usepackage{xcolor}
\usepackage{float}
\usepackage{accents}
\usepackage[export]{adjustbox}
\usepackage{setspace}
\usepackage{mathdots} 
\usepackage{empheq} 

	\usepackage[caption=false,font=normalsize,labelfont=sf,textfont=sf]{subfig}

\usepackage{pifont} 
\usepackage{caption}

\usepackage{booktabs, tabularx}
\usepackage{array}

\usepackage{amsthm}
\usepackage{enumitem}

\usepackage{tikz}
\usetikzlibrary{decorations.markings,chains,matrix}
\usetikzlibrary{intersections,calc}
\usepackage{tikz-3dplot} 
\usepackage{pgfplots}
\usepackage{pgfplotstable}
\pgfplotsset{compat = newest}
\usepackage{datatool}
\usepackage{textcomp}
\usetikzlibrary{shapes,arrows}
\usetikzlibrary{decorations.pathreplacing,angles,quotes}
\newtheorem{thm}{Theorem}
 \newtheorem{lemma}{Lemma}

\newtheorem{cor}{Corollary}
\newtheorem{defn}{Definition}

\newtheorem{rmk}{Remark}

\DeclareMathOperator{\FDP}{FDP}
\DeclareMathOperator{\FDR}{FDR}
\DeclareMathOperator{\TPP}{TPP}
\DeclareMathOperator{\TPR}{TPR}

\DeclareMathOperator{\Var}{Var}
\DeclareMathOperator{\fin}{fin}

\DeclareMathOperator{\sign}{sign}
\DeclareMathOperator{\thr}{thr}

\DeclareMathOperator{\DA}{DA}
\DeclareMathOperator{\SNR}{SNR}

\DeclareMathOperator{\Gr}{Gr}
\DeclareMathOperator{\cut}{c}
\DeclareMathOperator{\dist}{dist}

\newtheorem{innercustomgeneric}{\customgenericname}
\providecommand{\customgenericname}{}
\newcommand{\newcustomtheorem}[2]{%
  \newenvironment{#1}[1]
  {%
   \renewcommand\customgenericname{#2}%
   \renewcommand\theinnercustomgeneric{##1}%
   \innercustomgeneric
  }
  {\endinnercustomgeneric}
}

\newcustomtheorem{customthm}{Theorem}
\newcustomtheorem{customlemma}{Lemma}
\newcustomtheorem{customcor}{Corollary}
\newcustomtheorem{customassumption}{Assumption}

\newcommand{\y}{\boldsymbol{y}}
\newcommand{\x}{\boldsymbol{x}}

\newcommand{\X}{\boldsymbol{X}}

\newcommand{\bSigma}{\boldsymbol{\Sigma}}

\newcommand{\bbeta}{\boldsymbol{\beta}}
\newcommand{\bepsilon}{\boldsymbol{\epsilon}}
\newcommand{\hatbbeta}{\boldsymbol{\hat{\beta}}}

\newcommand{\A}{\mathcal{A}}

\newcommand{\XK}{\boldsymbol{\protect\accentset{\circ}{X}}}
\newcommand{\xK}{\boldsymbol{\protect\accentset{\circ}{x}}}
\newcommand{\XWK}{\boldsymbol{\widetilde{X}}}

\newcommand{\C}{\mathcal{C}}

\newcommand{\xKvar}{\protect\accentset{\circ}{x}}

\makeatletter
\let\save@mathaccent\mathaccent
\newcommand*\if@single[3]{%
  \setbox0\hbox{${\mathaccent"0362{#1}}^H$}%
  \setbox2\hbox{${\mathaccent"0362{\kern0pt#1}}^H$}%
  \ifdim\ht0=\ht2 #3\else #2\fi
  }
\newcommand*\rel@kern[1]{\kern#1\dimexpr\macc@kerna}
\newcommand*\widebar[1]{\@ifnextchar^{{\wide@bar{#1}{0}}}{\wide@bar{#1}{1}}}
\newcommand*\wide@bar[2]{\if@single{#1}{\wide@bar@{#1}{#2}{1}}{\wide@bar@{#1}{#2}{2}}}
\newcommand*\wide@bar@[3]{%
  \begingroup
  \def\mathaccent##1##2{%
    \let\mathaccent\save@mathaccent
    \if#32 \let\macc@nucleus\first@char \fi
    \setbox\z@\hbox{$\macc@style{\macc@nucleus}_{}$}%
    \setbox\tw@\hbox{$\macc@style{\macc@nucleus}{}_{}$}%
    \dimen@\wd\tw@
    \advance\dimen@-\wd\z@
    \divide\dimen@ 3
    \@tempdima\wd\tw@
    \advance\@tempdima-\scriptspace
    \divide\@tempdima 10
    \advance\dimen@-\@tempdima
    \ifdim\dimen@>\z@ \dimen@0pt\fi
    \rel@kern{0.6}\kern-\dimen@
    \if#31
      \overline{\rel@kern{-0.6}\kern\dimen@\macc@nucleus\rel@kern{0.4}\kern\dimen@}%
      \advance\dimen@0.4\dimexpr\macc@kerna
      \let\final@kern#2%
      \ifdim\dimen@<\z@ \let\final@kern1\fi
      \if\final@kern1 \kern-\dimen@\fi
    \else
      \overline{\rel@kern{-0.6}\kern\dimen@#1}%
    \fi
  }%
  \macc@depth\@ne
  \let\math@bgroup\@empty \let\math@egroup\macc@set@skewchar
  \mathsurround\z@ \frozen@everymath{\mathgroup\macc@group\relax}%
  \macc@set@skewchar\relax
  \let\mathaccentV\macc@nested@a
  \if#31
    \macc@nested@a\relax111{#1}%
  \else
    \def\gobble@till@marker##1\endmarker{}%
    \futurelet\first@char\gobble@till@marker#1\endmarker
    \ifcat\noexpand\first@char A\else
      \def\first@char{}%
    \fi
    \macc@nested@a\relax111{\first@char}%
  \fi
  \endgroup
}
\makeatother

\DeclareFontFamily{U}{dutchcal}{\skewchar\font=45}
\DeclareFontShape{U}{dutchcal}{m}{n}{<-> s*[1.2] dutchcal-r}{}
\DeclareFontShape{U}{dutchcal}{b}{n}{<-> s*[1.2] dutchcal-b}{}
\DeclareMathAlphabet{\mathdutchcal}{U}{dutchcal}{m}{n}
\SetMathAlphabet{\mathdutchcal}{bold}{U}{dutchcal}{b}{n}
\DeclareMathAlphabet{\mathdutchbcal}{U}{dutchcal}{b}{n}

\newlist{steps}{enumerate}{1}
\setlist[steps, 1]{label = {Step \arabic*:}, ref = {Step \arabic*}}

\newlist{alglist}{enumerate}{1}
\setlist[alglist, 1]{label = {\arabic*.}, ref = {\arabic*}}

\newcolumntype{L}[1]{>{\raggedright\arraybackslash}p{#1}}
\newcolumntype{C}[1]{>{\centering\arraybackslash}p{#1}}
\newcolumntype{R}[1]{>{\raggedleft\arraybackslash}p{#1}}

\definecolor{dark_green}{RGB}{102,166,30}
\definecolor{applegreen}{rgb}{0.55, 0.71, 0.0}

\definecolor{dark_red}{RGB}{217,95,2}
\definecolor{bittersweet}{rgb}{1.0, 0.44, 0.37}

\definecolor{dark_yellow}{RGB}{230,171,2}
\definecolor{bananayellow}{rgb}{1.0, 0.88, 0.21}

\mathtoolsset{showonlyrefs=true} 
\restylefloat{table}

\usepackage{xr}
\begin{document}

\title{High-Dimensional False Discovery Rate Control\\ for Dependent Variables}

\author{
Jasin Machkour, 
Michael Muma, 
and
Daniel P. Palomar
\thanks{J. Machkour,~Student Member,~IEEE, and M. Muma, Senior Member,~IEEE, are with the Robust Data Science Group at Technische Universit\"at Darmstadt, Germany (e-mail: jasin.machkour@tu-darmstadt.de; michael.muma@tu-darmstadt.de). D. P. Palomar,~Fellow,~IEEE, is with the Convex Optimization Group, The Hong Kong University of Science and Technology, Hong Kong SAR, China (e-mail: palomar@ust.hk).}
\thanks{J. Machkour is supported by the LOEWE initiative (Hesse, Germany) within the emergenCITY center. M. Muma is supported by the ERC Starting Grant ScReeningData (Project Number: 101042407). D.P. Palomar is supported by the Hong Kong GRF 16206123 research grant.}
\thanks{Extensive computations on the Lichtenberg High-Performance Computer of the Technische Universität Darmstadt were conducted for this research.}
}

\maketitle

\begin{abstract}
Algorithms that ensure reproducible findings from large-scale, high-dimensional data are pivotal in numerous signal processing applications. In recent years, multivariate false discovery rate (FDR) controlling methods have emerged, providing guarantees even in high-dimensional settings where the number of variables surpasses the number of samples. However, these methods often fail to reliably control the FDR in the presence of highly dependent variable groups, a common characteristic in fields such as genomics and finance. To tackle this critical issue, we introduce a novel framework that accounts for general dependency structures. Our proposed dependency-aware T-Rex selector integrates hierarchical graphical models within the T-Rex framework to effectively harness the dependency structure among variables. Leveraging martingale theory, we prove that our variable penalization mechanism ensures FDR control. We further generalize the FDR-controlling framework by stating and proving a clear condition necessary for designing both graphical and non-graphical models that capture dependencies. Additionally, we formulate a fully integrated optimal calibration algorithm that concurrently determines the parameters of the graphical model and the T-Rex framework, such that the FDR is controlled while maximizing the number of selected variables. Numerical experiments and a breast cancer survival analysis use-case demonstrate that the proposed method is the only one among the state-of-the-art benchmark methods that controls the FDR and reliably detects genes that have been previously identified to be related to breast cancer. An open-source implementation is available within the R package TRexSelector on CRAN.
\end{abstract}

\begin{IEEEkeywords}
T-Rex selector, false discovery rate (FDR) control, high-dimensional, martingale theory, survival analysis.
\end{IEEEkeywords}

\section{Introduction}
\label{sec: Introduction}
Reliably detecting the few relevant variables (i.e., features, biomarkers, or signals) in a large set of candidate variables given high-dimensional and noisy data is required in many of today's signal processing applications, e.g.~\cite{chung2007detection,chen2020false,chen2018sparse,tan2014direction,benidis2017sparse,zoubir2012robust, zoubir2018robust, machkour2017outlier, machkour2020robust, yang2019weakly,shen2007feature}. An important use-case of this work consists in detecting the few genes that are truly associated with the survival time of patients diagnosed with a certain type of cancer~\cite{tomczak2015review,Hosz2021survival,aalen2008survival}. The expression levels of the detected genes are then classified into low- and high-expressing genes, which allows cancer researchers to make statements such as: "The median survival time of breast cancer patients with a high expression of gene A and a low expression of gene B is $10$ years higher than the survival time of patients with a low expression of gene A and a high expression of gene B." Such information is invaluable for the development of new therapies and personalized medicine~\cite{kamel2017exploitation}.

However, the development and clinical trial of new drugs is costly and resources are limited. Therefore, it is crucial to select as many as possible of the few reproducible genes that are truly associated with the survival time of cancer patients while keeping the number of false discoveries (i.e., irrelevant genes) low. This aim is in line with false discovery rate (FDR) controlling methods. Such methods guarantee that the expected fraction of false discoveries among all discoveries (i.e., FDR) does not exceed a user-specified target level (e.g., $5$, $10$, $20$\%) while maximizing the number of selected variables. Popular FDR-controlling methods for the low-dimensional setting, where the number of samples $n$ is equal to or larger than the number of candidate variables $p$, are the Benjamini-Hochberg (\textit{BH}) method~\cite{benjamini1995controlling}, Benjamini-Yekutielli (\textit{BY}) method~\cite{benjamini2001control}, \textit{fixed-X} knockoff method~\cite{barber2015controlling}, and related approaches (e.g.,~\cite{gavrilov2009adaptive,barber2019knockoff}).

For the considered high-dimensional setting, where $p > n$, there exist \textit{model-X} knockoff methods~\cite{candes2018panning,sesia2019gene,romano2019deep,ren2021derandomizing} and Terminating-Random Experiments (\textit{T-Rex}) methods~\cite{machkour2021terminating,machkour2022TRexGVS,machkour2023ScreenTRex,scheidt2023FDRControlLaptop,machkour2023InformedEN,koka2024false,machkour2024sparse}.
It has been shown that the computational complexity of the \textit{T-Rex} selector~\cite{machkour2021terminating} is linear in $p$, which makes it scalable to millions of variables in a reasonable computation time, while the \textit{model-X} knockoff method~\cite{candes2018panning} is practically infeasible in such large-scale settings (see Figure~1 in~\cite{machkour2021terminating}). Unfortunately, however, both the \textit{model-X} knockoff methods and the \textit{T-Rex} methods fail to control the FDR reliably in the presence of groups of highly dependent variables, which are characteristic for, e.g., gene expression~\cite{segal2003regression}, genomics~\cite{balding2006tutorial}, and stock returns data~\cite{machkour2024TRexIndexTracking}.

In order to reduce the dependencies among the candidate variables, pruning approaches have been used~\cite{candes2018panning,sesia2019gene,machkour2021terminating}. In general, pruning methods cluster highly dependent variables into groups, select a representative variable for each group, and run the FDR controlling method on the set of representatives. This approach is suitable for genome-wide association studies (GWAS) based on large-scale high-dimensional genomics data from large biobanks~\cite{gwasCatalog,sudlow2015uk}, where the goal is to detect the groups of highly correlated single nucleotide polymorphisms (SNPs) that are associated with a disease of interest and not the specific SNPs. However, pruning methods are not applicable in gene expression analysis and other applications where it is crucial to detect specific genes or other variables.

Therefore, we propose a new FDR-controlling and dependency-aware \textit{T-Rex} (\textit{T-Rex+DA}) framework that provably controls the FDR at the user-specified target level. This is achieved and verified through the following theoretical contributions, numerical validations, and real world experiments:
\begin{enumerate}
\item A hierarchical graphical model (i.e., binary tree) is incorporated  into the \textit{T-Rex} framework and  is used to capture and leverage the dependency structure among variables to develop a variable penalization mechanism that allows for provable FDR control.
\item Using martingale theory~\cite{williams1991probability}, we prove that the proposed approach controls the FDR (Theorem~\ref{Theorem: Dependency-aware FDR control}).
\item We extend the proposed framework by stating and proving a comprehensible condition that must be satisfied for the design of graphical and non-graphical dependency-capturing models to be eligible for being incorporated into the \textit{T-Rex} framework (Theorem~\ref{Theorem: Group design principle}).
\item We develop a fully integrated optimal calibration algorithm that simultaneously determines the parameters of the incorporated graphical model and of the \textit{T-Rex} framework, such that the FDR is controlled while maximizing the number of selected variables (Theorem~\ref{Theorem: Optimal Dependency-Aware Calibration Algorithm}).
\item Numerical experiments and a real-world breast cancer survival analysis verify the theoretical results and demonstrate the practical usefulness of the proposed framework.
\end{enumerate}

Organization: Section~\ref{sec: The T-Rex Framework} briefly revisits the \textit{T-Rex} selector and introduces the FDR control problem. Section~\ref{sec: Methodology and Main Theoretical Results} analyzes the relative occurrences of correlated variables for the original \textit{T-Rex} framework (Theorem~\ref{Theorem: absolute difference relative occurrences}), describes the proposed \textit{T-Rex+DA} methodology, and proves our main theoretical results (Theorems~\ref{Theorem: Dependency-aware FDR control} and~\ref{Theorem: Group design principle}). Section~\ref{sec: Optimal Dependency-Aware T-Rex Algorithm} describes the implementation details and theoretical properties (Theorem~\ref{Theorem: Optimal Dependency-Aware Calibration Algorithm}) of the proposed \textit{T-Rex+DA} calibration algorithm. Section~\ref{sec: Numerical Experiments} numerically verifies the theoretical FDR control results and compares the proposed approach against state-of-the-art methods. Section~\ref{sec: FDR-Controlled Survival Analysis} presents the results of an FDR-controlled breast cancer survival analysis. Section~\ref{sec: Conclusion} concludes the paper.

\section{The T-Rex Framework}
\label{sec: The T-Rex Framework}
In this section, the FDR and the true positive rate (TPR) are defined and the \textit{T-Rex} framework is briefly revisited.

\subsection{FDR and TPR}
\label{subsec: FDR and TPR}
A general setting for sparse variable selection consists of $p = p_{1} + p_{0}$ variables out of which $p_{1}$ variables are true active variables (i.e., variables associated with a response of interest $\y$) and $p_{0}$ variables are null (i.e., non-active) variables. The task of sparse variable selection is to determine an estimator $\widehat{\A}$ of the true active set $\A \subseteq \lbrace 1, \ldots, p \rbrace$ of cardinality $| \A | = p_{1}$. In this general setting, the false discovery rate (FDR) and the true positive rate (TPR) are defined by
\begin{equation}
\FDR \coloneqq \mathbb{E} \big[ \FDP \big] \coloneqq \mathbb{E} \bigg [ \dfrac{| \widehat{\mathcal{A}} \backslash \mathcal{A}|}{| \widehat{\mathcal{A}} | \lor 1} \bigg ]
\label{eq: FDR}
\end{equation}
and 
\begin{equation}
\TPR \coloneqq \mathbb{E} \big[ \TPP \big] \coloneqq \mathbb{E} \bigg [ \dfrac{| \mathcal{A}  \cap \widehat{\mathcal{A}} |}{| \mathcal{A} | \lor 1} \bigg ],
\label{eq: TPR}
\end{equation}
respectively, where $\lor$ is the maximum operator, i.e., $| \widehat{\A} | \lor 1 \coloneqq \max \lbrace | \widehat{\A} |, 1 \rbrace$. In words,
\begin{enumerate}
\item the FDR is the expectation of the false discovery proportion (FDP), i.e., the fraction of selected null variables among all selected variables, and
\item the TPR is the expectation of the true positive proportion (TPP), i.e., the fraction of selected true active variables among all true active variables.
\end{enumerate}
In practice, a tradeoff between the FDR and TPR must be managed. That is, the TPR is maximized while not exceeding a user-specified target FDR level $\alpha \in [0, 1]$, i.e., $\FDR \leq \alpha$.

\subsection{The T-Rex Selector}
\label{subsec: The T-Rex Selector}
\begin{figure}[t]
\begin{center}
\scalebox{0.7}{
\hspace{-4em}
\begin{tikzpicture}[>=stealth]

  \coordinate (orig)   at (0,0);
  \coordinate (sample)   at (4,0.5);
  \coordinate (merge)   at (7,0.5);
  \coordinate (varSelect)   at (10,0.5);
  \coordinate (tFDR)   at (15.5,-1.1);
  \coordinate (fuse)   at (13,0.5);
  \coordinate (decision)   at (15.5,0.5);
  \coordinate (increment)   at (15.5,2.9);
  \coordinate (init)   at (7,-1.95);
  \coordinate (output)   at (19,0.5);
  
  \coordinate (between_scale_rank)   at (0.5,0.31);
  \coordinate (X_prime_to_tFDR_point)   at (0.5,6);
  \coordinate (X_prime_to_tFDR_point_point)   at (9,6);
  \coordinate (X_prime_to_merge_point)   at (0.5,-2.5);
  \coordinate (center_to_tFDR_point)   at (9.8,3.7);
  \coordinate (tFDR_to_fuse_point)   at (14.00,5.2);
  \coordinate (tFDR_to_sample_point)   at (4,5);
  
  \coordinate (decision_to_selection_1)   at (15.5,4.2);
  \coordinate (decision_to_selection_2)   at (10.2,4.2);

  \coordinate (Arrow_N_GenDummy)   at (4,3.06);
  \coordinate (Arrow_X_indVar)   at (7,3.06);
  \coordinate (Arrow_y_center)   at (9.5,3.7);
  \coordinate (Arrow_targetFDR_tFDR)   at (10,5.9);
  
   \coordinate (inference_Arrow)   at (15,-1.06);
   \coordinate (fuse_Arrow)   at (16.1,2.06);
  
  \coordinate (vdots2)   at (8.85,-0.1);
  \coordinate (vdots3)   at (11.5,-0.1);
  
  \coordinate (fuse_node)   at (15.00,0.5);
  
  \tikzstyle{decision} = [diamond, draw, 
										minimum width=2cm, minimum height=0.5cm, node distance=3cm, inner sep=0pt]

  \node[anchor=center, align=center] (D) at (merge) {};   
  \node[draw, minimum width=.7cm, minimum height=5.5cm, anchor=center, align=center] (E) at (varSelect) {\rotatebox{90}{\large Forward Variable Selection}};
  \node[draw, minimum width=.7cm, minimum height=5.5cm, anchor=center, align=center] (H) at (fuse) {\rotatebox{90}{\large Calibrate \& Fuse}};
  \node[decision, minimum width=2.7cm, minimum height=0.4cm, anchor=center, align=center] (M) at (decision) {\large $\widehat{\FDP} > \alpha$?};
  \node[draw, minimum width=2.5cm, minimum height=.7cm, anchor=center, align=center] (N) at (output) {\large Output: \\[0.3em] \large $\widehat{\mathcal{A}}_{L}(v^{*}, T^{*})$};
  \node[draw, minimum width=1.5cm, minimum height=.7cm, anchor=center, align=center] (O) at (increment) {\large $T \leftarrow T + 1$};
  \node (K) at (vdots2) {\large $\vdots$};
  \node (L) at (vdots3) {\large $\vdots$};
  
  \draw[->] ($(D.0) + (1,2.25)$) -- node[above]{$\y$} ($(E.0) + (-0.7,2.25)$);
  
  \draw[->] ($(D.0) + (1,1.0)$) -- node[above]{$\XWK_{1}$} ($(E.0) + (-0.7,1.0)$);
  \draw[->] ($(D.0) + (1,0.0)$) -- node[above]{$\XWK_{2}$} ($(E.0) + (-0.7,0.0)$);
  \draw[->] ($(D.0) + (1,-2.0)$) -- node[above]{$\XWK_{K}$} ($(E.0) + (-0.7,-2.0)$);
     
  \draw[->] ($(E.0) + (0,1.0)$) -- node[above]{$\C_{1, L}(T)$} ($(H.0) + (-0.7,1.0)$);
  \draw[->] ($(E.0) + (0,0.0)$) -- node[above]{$\C_{2, L}(T)$} ($(H.0) + (-0.7,0.0)$);
  \draw[->] ($(E.0) + (0,-2.0)$) -- node[above]{$\C_{K, L}(T)$} ($(H.0) +  (-0.7,-2.0)$);
 
  \coordinate (between_varSelect_fuse1)   at ($(E.0) + (2.75,1.5)$);
  \coordinate (between_varSelect_fuse2)   at ($(E.0) + (2.75,0.75)$);
  \coordinate (between_varSelect_fuse3)   at ($(E.0) + (2.75,-1.5)$);
 
  \draw[->] (tFDR) -- node[below, pos=0.1]{\large $\alpha$}(M);
  \draw[->] (H) -- (M);
  \draw[->] ($(M.0)$) -- node[above, pos=0.3]{Yes}(N);
  \draw[->] ($(M.90)$) -- node[right, pos=0.3]{No}(O);
  \path[draw,->] ($(O.90)$) -- (decision_to_selection_1) -- (decision_to_selection_2) --  ($(E.90) + (0.2,0)$);
  
\end{tikzpicture}}
\end{center}
\caption{Schematic overview of the \textit{T-Rex} framework.
}
\label{fig: T-Rex selector}
\end{figure}
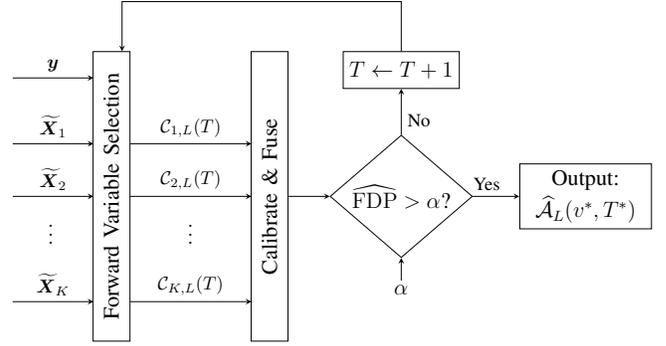

The \textit{T-Rex} selector is an FDR-controlling variable selection framework. A schematic overview of the generic framework is provided in Figure~\ref{fig: T-Rex selector}. A major characteristic of the \textit{T-Rex} framework is that it conducts $K$ supervised, independent, and early terminated random experiments based on a response vector $\y \in \mathbb{R}^{n}$ and the enlarged predictor matrices
\begin{equation}
\XWK_{k} = \big[ \X \,\, \XK_{k} \big] \in \mathbb{R}^{n \times (p + L)}, \, k = 1, \ldots, K,
\label{eq: enlarged predictor matrix}
\end{equation}
where $\X = [ \x_{1} \, \cdots \, \x_{p} ] \in \mathbb{R}^{n \times p}$ is the original predictor matrix containing the $p$ predictor variables and $\XK_{k} = [ \xK_{1} \, \cdots \, \xK_{L} ] \in\mathbb{R}^{n \times L}$ is a matrix containing $L$ dummy variables. The dummy vectors $\xK_{l} = [ \xKvar_{1, l} \, \cdots \, \xKvar_{n, l} ]^{\top}$, $l = 1, \ldots, L$, can be sampled independently from any univariate probability distribution with finite mean and variance (cf. Theorem~2 of~\cite{machkour2021terminating}). The $K$ early terminated random experiments are conducted with $\y$ and $\XWK_{k}$, $k = 1, \ldots, K$, as inputs to a forward variable selection method such as the \textit{LARS} algorithm~\cite{efron2004least}, \textit{Lasso}~\cite{tibshirani1996regression}, \textit{elastic net}~\cite{zou2005regularization}, or related methods (e.g.,~\cite{zou2006adaptive,machkour2020robust,machkour2017outlier}). These forward selection methods include no more than one variable in each iteration. All $K$ random experiments are terminated independently of each other after $T$ dummies have been included along their respective selection paths. This results in $K$ candidate sets $\C_{k, L}(T) \subseteq \lbrace 1, \ldots, p \rbrace$, $k = 1, \ldots, K$, which contain the indices of all original variables in $\X$ that have been selected before the random experiments have been terminated. Using the candidate sets, relative occurrences are computed for each variable, i.e.,
\begin{equation}
\Phi_{T,L}(j) \coloneqq 
\begin{cases}
\dfrac{1}{K} \sum\limits_{k = 1}^{K} \mathbbm{1}_{k}(j, T, L), & T \geq 1 \\
0, \quad & T = 0
\end{cases},
\label{eq: relative occurrences}
\end{equation}
where $\mathbbm{1}_{k}(j, T, L)$ is an indicator function that takes the value one if the $j$th variable is contained in the candidate set $\C_{k, L}(T)$ and zero otherwise. Based on these relative occurrences, the final set of selected variables is given by
\begin{equation}
\widehat{\mathcal{A}}_{L}(v, T) \coloneqq \lbrace j : \Phi_{T, L}(j) > v \rbrace,
\label{eq: selected active set}
\end{equation}
i.e., all variables whose relative occurrences exceed a voting threshold $v \in [0.5, 1)$ are selected. The extended calibration algorithm of the \textit{T-Rex} selector~\cite{machkour2021terminating} automatically determines the triple $(v^{*}, T^{*}, L) \in [0.5, 1) \times \lbrace 1, \ldots, L \rbrace \times \mathbb{N}_{+}$ such that the FDR is controlled at the user-defined target level $\alpha \in [0, 1]$, i.e.,
\begin{equation}
\FDR(v^{*}, T^{*}, L) \leq \alpha.
\label{eq: FDR controlled at target level}
\end{equation}
The calibrated \textit{T-Rex} parameters $v$ and $T$ are highlighted with a superscript ``$*$'' to emphasize that for any $L$, the parameters $v^{*}$ and $T^{*}$ are optimal in the sense that the FDR is controlled, while the number of selected variables is maximized (cf. Theorem~3 of~\cite{machkour2021terminating}). For the choice of the number of dummies $L$, the extended calibration algorithm considers a tradeoff between the computer memory consumption for storing large dummy matrices and maximizing the number of selected variables. However, the FDR is always controlled for any choice of $L$ (cf. Theorem~1 of~\cite{machkour2021terminating}). The extended calibration algorithm achieves the optimal solution $(v^{*}, T^{*})$ by
\begin{enumerate}
\item first terminating all random experiments after only one dummy has been included,
\item computing a conservative estimator $\widehat{\FDP}(v, T, L)$ that satisfies
\begin{equation}
\FDR(v, T, L) = \mathbb{E} \big[ \FDP(v, T, L) \big] \leq \mathbb{E} \big[ \widehat{\FDP}(v, T, L) \big],
\label{eq: conservative FDP estimator}
\end{equation}
\item iteratively increasing the number of included dummies $T$ until $\widehat{\FDP}(v = 1 - 1/K, T, L)$ (i.e., the FDP estimator at the effectively highest voting level $v = 1 - 1/K$) exceeds the target level $\alpha$ for the first time, and
\item returning to the solution $(v^{*}, T^{*})$ of the preceding iteration $T^{*}$ that satisfies the equation
\begin{equation}
v^{*} = \inf \big\lbrace \nu \in [0.5, 1) : \widehat{\FDP}(\nu, T^{*}, L) \leq \alpha \big\rbrace,
\label{eq: stopping time v for T-Rex}
\end{equation}
where $v^{*}$ is the lowest feasible voting level such that $\widehat{\FDP}$ does not exceed the target level $\alpha$.
\end{enumerate}
For details on the design and properties of the conservative FDP estimator $\widehat{\FDP}$, we refer the interested reader to~\cite{machkour2021terminating}.

\section{Methodology and Main Theoretical Results}
\label{sec: Methodology and Main Theoretical Results}
In this section, the proposed FDR-controlling \textit{T-Rex+DA} framework for general dependency structures is introduced. First, a dependency-capturing graph model is incorporated into the \textit{T-Rex+DA} framework. Second, we prove that the considered group design yields FDR control. Finally, we formulate a sufficient group design condition for graphical as well as non-graphical models that can be used as a guiding principle for other application-specific group designs.

\subsection{Preliminaries}
\label{subsec: Preliminaries}
Before the proposed \textit{T-Rex+DA} selector is presented and in order to understand why the ordinary \textit{T-Rex} selector might loose the FDR control property in the presence of highly correlated variables, we establish an interesting relationship between the pairwise relative occurrences of two candidate variables and the correlation coefficient between them. For example, let the \textit{Lasso} \cite{tibshirani1996regression} be used to perform the forward variable selection in each random experiment. Within the \textit{T-Rex} framework and for the $k$th random experiment, the \textit{Lasso} estimator is defined by
\begin{equation}
\hatbbeta_{k}(\lambda_{k}(T, L)) = \underset{\bbeta_{k}}{\arg\min} \, \dfrac{1}{2} \big\| \y - \XWK_{k}\bbeta_{k} \big\|_{2}^{2} + \lambda_{k}(T, L) \cdot \| \bbeta_{k} \|_{1},
\label{eq: Lasso for the kth random experiment}
\noeqref{eq: Lasso for the kth random experiment}
\end{equation}
where $\lambda_{k}(T, L) > 0$ is the sparsity parameter that corresponds to the change point in the $k$th random experiment after $T$ dummies have been included. With these definitions in place we can formulate the following theorem:

\begin{thm}[Absolute difference of relative occurrences]
Let $\rho_{j, j^{\prime}} \coloneqq \x_{j}^{\top} \x_{j^{\prime}}$, $j, j^{\prime} \in \lbrace 1, \ldots,p \rbrace$, be the sample correlation coefficient of the standardized variables $j$ and $j^{\prime}$. Suppose that $\hat{\beta}_{j, k}, \hat{\beta}_{j^{\prime}, k} \neq 0$. Then, for all tuples $(T, L) \in \lbrace 1, \ldots, L \rbrace \times \mathbb{N}_{+}$ it holds that
\begin{equation}
\big| \Phi_{T, L}(j) - \Phi_{T, L}(j^{\prime}) \big| \leq \widebar{\Lambda} \| \y \|_{2} \cdot \sqrt{2 (1 - \rho_{j, j^{\prime}})},
\label{eq: Theorem - absolute difference relative occurrences}
\end{equation}
where $\widebar{\Lambda} \coloneqq \frac{1}{K} \sum_{k = 1}^{K} \frac{1}{\lambda_{k}(T, L)}$.
\label{Theorem: absolute difference relative occurrences}
\end{thm}

\begin{proof}
The proof is deferred to Appendix~\ref{appendix: Proofs and Technical Lemmas}.
\label{proof: Theorem - absolute difference relative occurrences}
\end{proof}

\subsection{Proposed: The Dependency-Aware T-Rex Selector}
\label{subsec: Proposed: The Dependency-Aware T-Rex Selector}
From Theorem~\ref{Theorem: absolute difference relative occurrences}, we know that the pairwise absolute differences between the relative occurrences are bounded and the differences are zero when the corresponding variables are perfectly correlated. That is, even if only one of the variables from the pair of highly correlated variables is a true active variable, both variables might be selected. This is the harmful behavior that leads to the loss of the FDR control property in the presence of highly correlated variables. Loosely speaking, if a candidate variable is highly correlated with another candidate variable and has a similar relative occurrence, then even high relative occurrences are no evidence for that variable being a true active one. Therefore, for such types of data, we propose to replace the ordinary relative occurrences of the \textit{T-Rex} selector $\Phi_{T, L}(j)$ by the dependency-aware relative occurrences $\Phi_{T, L}^{\DA}(j, \rho_{\thr})$, $j = 1, \ldots, p$, which are defined as follows:
\begin{defn}[Dependency-aware relative occurrences]
The dependency-aware relative occurrence of variable $j \in \lbrace 1, \ldots, p \rbrace$ is defined by
\begin{equation}
\Phi_{T, L}^{\DA}(j, \rho_{\thr}) \coloneqq \Psi_{T, L}(j, \rho_{\thr}) \cdot \Phi_{T, L}(j),
\label{eq: dependency-aware relative occurrences}
\end{equation}
where
\begin{align}
&\Psi_{T, L}(j, \rho_{\thr}) \coloneqq
\label{eq: dependency factor}
\\
&
\begin{cases}
 \dfrac{1}{2 \,\,\, - \,\,\, \smashoperator{\min\limits_{j^{\prime} \in \Gr(j, \rho_{\thr})}} \quad\, \big\lbrace \big| \Phi_{T, L}(j) - \Phi_{T, L}(j^{\prime}) \big| \big\rbrace}, &\!\! \Gr(j, \rho_{\thr}) \neq \varnothing
 \\[1em]
 1 / 2, &\!\! \Gr(j, \rho_{\thr}) = \varnothing
 \end{cases},
\end{align}
with $\Psi_{T, L}(j, \rho_{\thr}) \in [0.5, 1]$, is a penalty factor,
\begin{equation}
\Gr(j, \rho_{\thr}) \subseteq \lbrace 1, \ldots, p \rbrace \backslash \lbrace j \rbrace
\label{eq: generic definition of variable groups}
\end{equation}
is the generic definition of the group of variables that are associated with variable $j$, and $\rho_{\thr} \in [0, 1]$ is a parameter that determines the size of the variable groups.
\label{definition: Dependency-aware relative occurrences}
\end{defn}
In words, the dependency-aware relative occurrence of variable $j$ is designed to penalize the ordinary relative occurrence of variable $j$ according to its resemblance with the relative occurrences of its associated group of variables $\Gr(j, \rho_{\thr})$.

From Definition~\ref{definition: Dependency-aware relative occurrences}, we can infer that the selected active set of the proposed \textit{T-Rex+DA} selector is a subset of the selected active set of the ordinary \textit{T-Rex} selector in~\eqref{eq: selected active set}:
\begin{cor}
Let $\widehat{\A}_{L}(v, T) \coloneqq \lbrace j : \Phi_{T, L}(j) > v \rbrace$ and $\widehat{\A}_{L}(v, \rho_{\thr}, T) \coloneqq \lbrace j : \Phi_{T, L}^{\DA}(j, \rho_{\thr}) > v \rbrace$ be the selected active sets of the ordinary \textit{T-Rex} selector and the \textit{T-Rex+DA} selector, respectively. Then, it holds that
\begin{equation}
\widehat{\A}_{L}(v, \rho_{\thr}, T) \subseteq \widehat{\A}_{L}(v, T).
\label{eq: T-Rex+DA solution is subset of T-Rex solution}
\end{equation}
\label{Corollary: T-Rex+DA solution is subset of T-Rex solution}
\end{cor}

\begin{proof}
Using the definition of $\Phi_{T, L}^{\DA}(j, \rho_{\thr})$ in~\eqref{eq: dependency-aware relative occurrences}, we obtain
\begin{align}
\widehat{\A}_{L}(v, \rho_{\thr}, T) &= \lbrace j :  \Psi_{T, L}(j, \rho_{\thr}) \cdot \Phi_{T, L}(j) > v \rbrace
\label{proof: Corollary - T-Rex+DA solution is subset of T-Rex solution - 1}
\\
&\subseteq \lbrace j : \Phi_{T, L}(j) > v \rbrace
\label{proof: Corollary - T-Rex+DA solution is subset of T-Rex solution - 2}
\\
&=\widehat{\A}_{L}(v, T),
\label{proof: Corollary - T-Rex+DA solution is subset of T-Rex solution - 3}
\end{align}
where the second line follows from $\Psi_{T, L}(j, \rho_{\thr}) \leq 1$.
\label{proof: Corollary - T-Rex+DA solution is subset of T-Rex solution}
\end{proof}
Loosely speaking, Corollary~\ref{Corollary: T-Rex+DA solution is subset of T-Rex solution} indicates that the effect of replacing $\Phi_{T, L}(j)$ by $\Phi_{T, L}^{\DA}(j, \rho_{\thr})$ is that highly correlated variables, for which there is not sufficient evidence to decide if they are active, are removed from the selected active set.

In order to particularize the \textit{T-Rex+DA} selector for different dependency structures among the candidate variables, only the generic definition of the variable groups $\Gr(j, \rho_{\thr})$ in~\eqref{eq: generic definition of variable groups} has to be specified. Therefore, this work develops a rigorous methodology for the design of $\Gr(j, \rho_{\thr})$ such that the FDR is provably controlled at the user-defined target level $\alpha$ while maximizing the number of selected variables and, thus, implicitly maximizing the TPR.

\subsection{Clustering Variables via Hierarchical Graphical Models}
\label{subsec: Clustering Variables via Hierarchical Graphical Models}
In the following, we specify $\Gr(j, \rho_{\thr})$, $j = 1, \ldots, p$, using a hierarchical graphical model. That is, the variables $\x_{1}, \ldots, \x_{p}$ are clustered in a recursive fashion according to some measure of distance. The resulting binary tree or dendrogram is a structured graph that allows for different distance cutoff values that partition the set of variables. Figure~\ref{fig: dendrogram} depicts such a dendrogram for $p = 6$ variables, where the height of the ``$\sqcap$''-shaped connector of any two clusters represents the distance of the two connected clusters. At the bottom of the dendrogram, all variables are considered as one-element clusters. Then, starting at the bottom, in each iteration the two clusters with the smallest distance are connected until all variables are clustered into a single cluster at the top. The obtained dendrogram can be evaluated at different distances (i.e., values on the $y$-axis), resulting in different variable clusters. The $p$ discrete distances between two consecutive cutoff levels that invoke a change in the clusters, are denoted by $\Delta\rho_{\thr, u}$, $u = 1, \ldots, p$. For example, cutting off the dendrogram in Figure~\ref{fig: dendrogram} at a distance of
\begin{align}
1 - \rho_{\thr}(u_{\cut} = 2) &\coloneqq 1 - \smashoperator{\sum\limits_{u = 1}^{u_{\cut} = 2}} \Delta\rho_{\thr, u} 
\\
&= 1 - (0.1 + 0.3) = 0.6,
\label{eq: cutoff dendrogram}
\end{align}
where $u_{\cut} \in \lbrace 1, \ldots, p \rbrace$ is the discrete cutoff level, yields three disjoint variable clusters: $\lbrace \x_{1}, \x_{2} \rbrace$, $\lbrace \x_{3}, \x_{4}, \x_{5} \rbrace$, and $\lbrace \x_{6} \rbrace$.

With this generic description of hierarchical graphical models in place, we can specify the generic definition of the variable groups in~\eqref{eq: generic definition of variable groups} in a recursive fashion:
\begin{defn}[Hierarchical group design]
The $j$th variable group following a hierarchical graphical model (i.e., binary tree/dendrogram) is defined by
\begin{align}
\Gr&(j, \rho_{\thr}(u_{\cut})) \coloneqq
\\
\big\lbrace & j^{\prime} \in \lbrace 1, \ldots, p \rbrace \backslash \lbrace j \rbrace :
\label{eq: recursive formulation group of correlated variables}
\\
&\,\, \dist_{u_{\cut} - 1}(j, j^{\prime}) \in [ 1 - \rho_{\thr}(u_{\cut}), 1 - \rho_{\thr}(u_{\cut} - 1) ] \big\rbrace,
\end{align}
where $\dist_{u_{\cut} - 1}(j, j^{\prime})$ is a still to be specified measure of distance between the groups $\Gr(j, \rho_{\thr}(u_{\cut} - 1))$ and $\Gr(j^{\prime}, \rho_{\thr}(u_{\cut} - 1))$.
\label{Definition: binary tree group model}
\end{defn}
Note that in this recursive definition of the variable groups, we consider $\rho_{\thr}(u_{\cut})$ to be a variable that can be optimized and, therefore, include it in $\Gr(j, \rho_{\thr}(u_{\cut}))$ as a second argument.
%
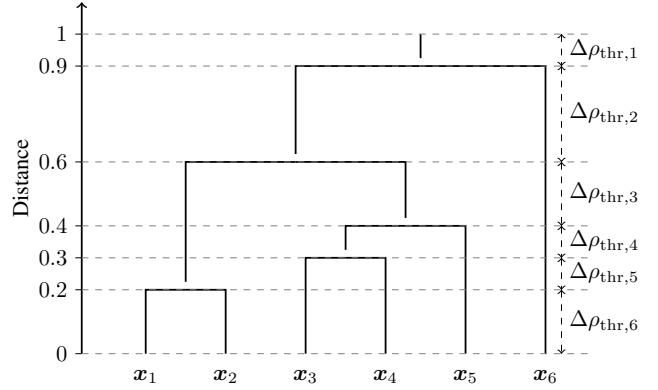
\begin{figure}[t]
\begin{center}
\scalebox{0.85}{
\begin{tikzpicture}[sloped]
    \node[label=below:$\x_{1}$] (x1) at (-6,0) {};
    \node[label=below:$\x_{2}$] (x2) at (-4.75,0) {};
    \node[label=below:$\x_{3}$] (x3) at (-3.5,0) {};
    \node[label=below:$\x_{4}$] (x4) at (-2.25,0) {};
    \node[label=below:$\x_{5}$] (x5) at (-1,0) {};
    \node[label=below:$\x_{6}$] (x6) at (0.25,0) {};
    
    \node (x12) at (-5.375,1) {};
    \node (x34) at (-2.875,1.5) {};
    \node (x345) at (-1.9375,2) {};
    \node (x12345) at (-3.65625,3) {};
    \node (x123456) at (-1.703125,4.5) {};
    \node (end) at (-1.703125,5) {};
	
    \draw[line width=0.8pt]  (x1.center) |- (x12.center);
    \draw[line width=0.8pt]  (x2.center) |- (x12.center);
    \draw[line width=0.8pt]  (x3.center) |- (x34.center);
    \draw[line width=0.8pt]  (x4.center) |- (x34.center);
    \draw[line width=0.8pt]  (x34) |- (x345.center);
    \draw[line width=0.8pt]  (x5.center) |- (x345.center);
    \draw[line width=0.8pt]  (x12) |- (x12345.center);
    \draw[line width=0.8pt]  (x345) |- (x12345.center);
    \draw[line width=0.8pt]  (x12345) |- (x123456.center);
    \draw[line width=0.8pt]  (x6.center) |- (x123456.center);
    \draw[line width=0.8pt]  (x123456) |- (end.center);
    
    \draw[->,line width=0.8pt] (-7,0) -- node[above,yshift=2em]{Distance}(-7,5.5);
    
    \foreach \y in {0,1,1.5,2,3,4.5,5}                     
    \draw[shift={(0,\y)},color=black] (-7,0) -- (-7.1,0);
    
    \node[left] at (-7.1,0) {$0$} ;
    \node[left] at (-7.1,1) {$0.2$} ;
     \node[left] at (-7.1,1.5) {$0.3$} ;
    \node[left] at (-7.1,2) {$0.4$} ;
    \node[left] at (-7.1,3) {$0.6$} ;
     \node[left] at (-7.1,4.5) {$0.9$} ;
    \node[left] at (-7.1,5) {$1$} ;
	
	\draw[dashed,color=gray] (-7,5) -- (1,5);
	\draw[<->,dashed,color=black] (0.5,5) -- node[right,rotate=90]{$\Delta\rho_{\thr, 1}$}(0.5,4.5);
    \draw[dashed,color=gray] (-7,4.5) -- (1,4.5);
    \draw[<->,dashed,color=black] (0.5,4.5) -- node[right,rotate=90]{$\Delta\rho_{\thr, 2}$}(0.5,3);
    \draw[dashed,color=gray] (-7,3) -- (1,3);
    \draw[<->,dashed,color=black] (0.5,3) -- node[right,rotate=90]{$\Delta\rho_{\thr, 3}$}(0.5,2);
    \draw[dashed,color=gray] (-7,2) -- (1,2);
    \draw[<->,dashed,color=black] (0.5,2) -- node[right,rotate=90]{$\Delta\rho_{\thr, 4}$}(0.5,1.5);
	\draw[dashed,color=gray] (-7,1.5) -- (1,1.5);
	\draw[<->,dashed,color=black] (0.5,1.5) -- node[right,rotate=90]{$\Delta\rho_{\thr, 5}$}(0.5,1);
	\draw[dashed,color=gray] (-7,1) -- (1,1);
	\draw[<->,dashed,color=black] (0.5,1) -- node[right,rotate=90]{$\Delta\rho_{\thr, 6}$}(0.5,0);
	\draw[dashed,color=gray] (-7,0) -- (1,0);
\end{tikzpicture}}
\end{center}
\caption{Hierarchical graphical models: The dendrogram.}
\label{fig: dendrogram}
\end{figure}
%

\begin{rmk}
The following three distance measures are frequently used in hierarchical graphical models~\cite{murtagh2012algorithms}:
\begin{enumerate}[wide, labelwidth=!, labelindent=0pt]
\item \underline{Single linkage}:
\begin{equation}
\dist_{u_{\cut}}(g, h) \coloneqq \min_{
\begin{array}{l}
g^{\prime} \in \Gr(g, \rho_{\thr}(u_{\cut}))
\\
h^{\prime} \in \Gr(h, \rho_{\thr}(u_{\cut}))
\end{array}
} 1 - | \rho_{g^{\prime}, h^{\prime}} |,
\label{eq: single linkage distance}
\end{equation}
\item \underline{Complete linkage}:
\begin{equation}
\dist_{u_{\cut}}(g, h) \coloneqq \max_{
\begin{array}{l}
g^{\prime} \in \Gr(g, \rho_{\thr}(u_{\cut}))
\\
h^{\prime} \in \Gr(h, \rho_{\thr}(u_{\cut}))
\end{array}
} 1 - | \rho_{g^{\prime}, h^{\prime}} |,
\label{eq: complete linkage distance}
\end{equation}
\item \underline{Average linkage}:
\begin{equation}
\dist_{u_{\cut}}(g, h) \coloneqq 
\dfrac{
\qquad
\smashoperator{\sum\limits_{
\footnotesize{
\begin{array}{c}
g^{\prime} \in
\\
\Gr(g, \rho_{\thr}(u_{\cut}))
\end{array}
}
}}
\qquad\qquad\quad
\smashoperator{\sum\limits_{
\footnotesize{
\begin{array}{c}
h^{\prime} \in
\\
\Gr(h, \rho_{\thr}(u_{\cut}))
\end{array}
}
}}
\big( 1 - | \rho_{g^{\prime}, h^{\prime}} | \big)}{| \Gr(g, \rho_{\thr}(u_{\cut})) | \cdot |  \Gr(h, \rho_{\thr}(u_{\cut})) |}.
\label{eq: average linkage distance}
\end{equation}
\end{enumerate}
\label{Remark: distance measures for dendrogram}
\end{rmk}

\begin{rmk}
Note that, for all $u_{\cut} \in \lbrace 1, \ldots, p \rbrace$, it holds that
\begin{align}
\Gr&(j_{1}, \rho_{\thr}(u_{\cut})) \cap \Gr(j_{2}, \rho_{\thr}(u_{\cut})) =
\\[0.5em]
&
\begin{cases}
\varnothing, \qquad \nexists \, j \in \lbrace 1, \ldots, p \rbrace : j_{1}, j_{2} \in \Gr(j, \rho_{\thr}(u_{\cut}))
\\[1em]
\Gr(j_{1}, \rho_{\thr}(u_{\cut})) \cup \Gr(j_{2}, \rho_{\thr}(u_{\cut})), \quad \text{otherwise}
\end{cases},
\label{eq: identical or disjoint variable groups}
\end{align}
i.e., any arbitrary pair of groups among the obtained $p$ groups $\Gr(j, \rho_{\thr}(u_{\cut}))$, $j = 1, \ldots, p$, are disjoint if and only if there exist no two variables $j_{1}$ and $j_{2}$ that belong to the same group and identical otherwise. Loosely speaking, due to the binary tree structure of hierarchical graphical models, there exist no ``overlapping'' variable groups.
\label{Remark: no overlapping groups}
\end{rmk}

\subsection{Preliminaries for the FDR Control Theorem}
\label{subsec: Preliminaries for the FDR Control Theorem}
Based on the recursive definition of the variable groups in~\eqref{eq: recursive formulation group of correlated variables}, we can formulate the FDR, TPR, and the conservative FDP estimator $\widehat{\FDP}$. For this purpose, let
\begingroup
\begin{align}
V_{T,L}(v, \rho_{\thr}(u_{\cut})) &\coloneqq \big| \widehat{\mathcal{A}}^{\, 0}(v, \rho_{\thr}(u_{\cut})) \big|
\\
&\coloneqq  \big| \lbrace \text{null } j : \Phi_{T, L}^{\DA}(j, \rho_{\thr}(u_{\cut})) > v \rbrace  \big|, 
\label{eq: number of selected null variables}
\\[1em]
S_{T,L}(v, \rho_{\thr}(u_{\cut})) &\coloneqq \big| \widehat{\mathcal{A}}^{\, 1}(v, \rho_{\thr}(u_{\cut})) \big| \\
&\coloneqq  \big| \lbrace \text{active } j : \Phi_{T, L}^{\DA}(j, \rho_{\thr}(u_{\cut})) > v \rbrace  \big|,
\label{eq: number of selected active variables}
\\[1em]
R_{T,L}(v, \rho_{\thr}(u_{\cut})) &\coloneqq \big| \widehat{\mathcal{A}}(v, \rho_{\thr}(u_{\cut})) \big| \\
&\coloneqq  \big| \lbrace j : \Phi_{T, L}^{\DA}(j, \rho_{\thr}(u_{\cut})) > v \rbrace  \big|,
\label{eq: number of selected variables}
\end{align}
\endgroup
be the number of selected null variables, the number of selected active variables, and the total number of selected variables, respectively. Note that the expressions $\widehat{\mathcal{A}}^{\, 0}(v, \rho_{\thr}(u_{\cut}))$, $\widehat{\mathcal{A}}^{\, 1}(v, \rho_{\thr}(u_{\cut}))$, and $\widehat{\mathcal{A}}(v, \rho_{\thr}(u_{\cut}))$ are shortcuts (i.e., $L$ and $T$ are dropped) of the expressions $\widehat{\mathcal{A}}_{L}^{\, 0}(v, \rho_{\thr}(u_{\cut}), T)$, $\widehat{\mathcal{A}}_{L}^{\, 1}(v, \rho_{\thr}(u_{\cut}), T)$, and $\widehat{\mathcal{A}}_{L}(v, \rho_{\thr}(u_{\cut}), T)$, respectively.

\begin{defn}[Dependency-aware FDP and FDR]
The dependency-aware FDR is defined as the expectation of the dependency-aware FDP, i.e.,
\begin{align}
\FDR(v, \rho_{\thr}(u_{\cut}), T, L) &\coloneqq \mathbb{E} \big[ \FDP(v, \rho_{\thr}(u_{\cut}), T, L) \big] 
\\[0.75em]
&\coloneqq \mathbb{E} \bigg[ \dfrac{V_{T, L}(v, \rho_{\thr}(u_{\cut}))}{R_{T, L}(v, \rho_{\thr}(u_{\cut})) \lor 1} \bigg].
\label{eq: dependency-aware FDP and FDR}
\end{align}
\label{definition: dependency-aware FDP and FDR}
\end{defn}
\begin{defn}[Dependency-aware TPP and TPR]
The dependency-aware TPR is defined as the expectation of the dependency-aware TPP, i.e.,
\begin{align}
\TPR(v, \rho_{\thr}(u_{\cut}), T, L) &\coloneqq \mathbb{E} \big[ \TPP(v, \rho_{\thr}(u_{\cut}), T, L) \big] 
\\[0.75em]
&\coloneqq \mathbb{E} \bigg[ \dfrac{S_{T, L}(v, \rho_{\thr}(u_{\cut}))}{p_{1} \lor 1} \bigg].
\label{eq: dependency-aware TPP and TPR}
\end{align}
\label{definition: dependency-aware TPP and TPR}
\end{defn}

From Definition~\ref{definition: dependency-aware FDP and FDR}, we know that in order to design a dependency-aware and conservative FDP estimator, we only need to design a dependency-aware estimator of the number of selected null variables $V_{T, L}(v, \rho_{\thr}(u_{\cut}))$, since the total number of selected variables $R_{T, L}(v, \rho_{\thr}(u_{\cut}))$ is observable. For this purpose, we plug the dependency-aware relative occurrences from Definition~\ref{definition: Dependency-aware relative occurrences} and the group design from Definition~\ref{Definition: binary tree group model} into the ordinary \textit{T-Rex} estimator of $V_{T, L}(v)$ in~\cite{machkour2021terminating}, which yields the dependency-aware estimator of the number of selected null variables

\begingroup
\allowdisplaybreaks
\begin{align}
&\widehat{V}_{T, L}(v, \rho_{\thr}(u_{\cut})) \coloneqq
\label{eq: estimator of V_T_L with hierarchical group design}
\\
&\qquad\qquad \smashoperator{\sum\limits_{j \in \widehat{\A}(v, \rho_{\thr}(u_{\cut}))}} \Big( 1 - \Phi_{T, L}^{\DA}(j, \rho_{\thr}(u_{\cut})) \Big) + \widehat{V}_{T, L}^{\prime}(v, \rho_{\thr}(u_{\cut})),
\end{align}
where
\begin{align}
&\widehat{V}_{T, L}^{\prime}(v, \rho_{\thr}(u_{\cut})) \coloneqq
\label{eq: V_T_L_prime}
\\
&\, \sum\limits_{t = 1}^{T} \dfrac{p - \sum\limits_{q = 1}^{p}\Phi_{t, L}^{\DA}(q, \rho_{\thr}(u_{\cut}))}{L - (t - 1)} \cdot \dfrac{\qquad\, \smashoperator{\sum\limits_{j \in \widehat{\A}(v, \rho_{\thr}(u_{\cut}))}} \Delta\Phi_{t, L}^{\DA}(j, \rho_{\thr}(u_{\cut}))}{\qquad\, \smashoperator{\sum\limits_{j \in \widehat{\A}(0.5, \rho_{\thr}(u_{\cut}))}} \Delta\Phi_{t, L}^{\DA}(j, \rho_{\thr}(u_{\cut}))}
\end{align}
\endgroup
and $\Delta\Phi_{t, L}^{\DA}(j, \rho_{\thr}) \coloneqq \Phi_{t, L}^{\DA}(j, \rho_{\thr}) - \Phi_{t - 1, L}^{\DA}(j, \rho_{\thr})$ is the increase in the dependency-aware relative occurrence from step $t - 1$ to $t$. The expressions in~\eqref{eq: estimator of V_T_L with hierarchical group design} and~\eqref{eq: V_T_L_prime} are derived along the lines of the ordinary estimator of $V_{T, L}(v)$ in~\cite{machkour2021terminating} except that the ordinary relative occurrences have been replaced by the proposed dependency-aware relative occurrences in Definition~\ref{definition: Dependency-aware relative occurrences}. Thus, for more details on the derivation of~\eqref{eq: estimator of V_T_L with hierarchical group design} and~\eqref{eq: V_T_L_prime}, we refer the interested reader to~\cite{machkour2021terminating}.

Finally, the conservative estimator of the FDP is defined as follows:
\begin{defn}[Dependency-aware FDP estimator]
The dependency-aware FDP estimator is defined by
\begin{equation}
\widehat{\FDP}(v, \rho_{\thr}(u_{\cut}), T, L) \coloneqq \dfrac{\widehat{V}_{T, L}(v, \rho_{\thr}(u_{\cut}))}{R_{T, L}(v, \rho_{\thr}(u_{\cut})) \lor 1}.
\label{eq: dependency-aware FDP estimator}
\end{equation}
\label{definition: dependency-aware FDP estimator}
\end{defn}

With all preliminary definitions in place, the overarching goal of this paper, i.e., maximizing the number of selected variables while controlling the FDR at the target level $\alpha$, is formulated as follows:
\begin{empheq}[box=\fbox]{align}
&\max_{v, \rho_{\thr}(u_{\cut}), T} \, R_{T, L}(v, \rho_{\thr}(u_{\cut}))
\\
&\qquad\text{subject  to} \quad \widehat{\FDP}(v, \rho_{\thr}(u_{\cut}), T, L) \leq \alpha.
\label{eq: optimization problem - FDR-controlled variable selection}
\end{empheq}
In Section~\ref{subsec: Dependency-Aware FDR Control}, we prove that satisfying the condition of the optimization problem in~\eqref{eq: optimization problem - FDR-controlled variable selection} yields FDR control and in Section~\ref{sec: Optimal Dependency-Aware T-Rex Algorithm}, we propose an efficient algorithm to solve~\eqref{eq: optimization problem - FDR-controlled variable selection} and prove that it yields an optimal solution.

\subsection{Dependency-Aware FDR Control}
\label{subsec: Dependency-Aware FDR Control}
In this section, using martingale theory~\cite{williams1991probability} we state and prove that controlling $\widehat{\FDP}(v, \rho_{\thr}(u_{\cut}), T, L)$ at the target level $\alpha$ (i.e., the condition in the optimization problem in~\eqref{eq: optimization problem - FDR-controlled variable selection}) guarantees FDR control. The required technical Lemmas~\ref{Lemma: T-Rex+DA super-martingale},~\ref{Lemma: V_T,L monotonically increasing in rho}, and~\ref{Lemma: V'_hat_T,L monotonically decreasing in rho} are deferred to Appendix~\ref{appendix: Proofs and Technical Lemmas}.
\begin{thm}[Dependency-aware FDR control]
For all quadruples $(T, L, \rho_{\thr}(u_{\cut}), v) \in \lbrace 1, \ldots, L \rbrace \times \mathbb{N}_{+} \times [0, 1] \times [0.5, 1)$ that satisfy the equation
\begin{equation}
v = \inf \big\lbrace \nu \in [0.5, 1) : \widehat{\FDP}(\nu, \rho_{\thr}(u_{\cut}), T, L) \leq \alpha \big\rbrace,
\label{eq: stopping time v}
\end{equation}
and as $K \rightarrow \infty$, the \textit{T-Rex+DA} selector with $\Gr(j, \rho_{\thr}(u_{\cut}))$ from Definition~\ref{Definition: binary tree group model} controls the FDR at any fixed target level $\alpha \in [0, 1]$, i.e.,
\begin{equation}
\FDR(v, \rho_{\thr}(u_{\cut}), T, L) \leq \alpha.
\label{eq: dependency-aware FDR control}
\end{equation}
\label{Theorem: Dependency-aware FDR control}
\end{thm}
\begin{proof}
Rewriting the expression for the FDP in Definition~\ref{definition: dependency-aware FDP and FDR}, we obtain
\begin{align}
& \FDP(v, \rho_{\thr}(u_{\cut}), T, L)
\\
& \quad = \dfrac{V_{T, L}(v, \rho_{\thr}(u_{\cut}))}{R_{T, L}(v, \rho_{\thr}(u_{\cut})) \lor 1}
\\
& \quad = \dfrac{\widehat{V}_{T, L}(v, \rho_{\thr}(u_{\cut}))}{R_{T, L}(v, \rho_{\thr}(u_{\cut})) \lor 1} \cdot \dfrac{V_{T, L}(v, \rho_{\thr}(u_{\cut}))}{\widehat{V}_{T, L}(v, \rho_{\thr}(u_{\cut}))}
\\
& \quad \leq \alpha \cdot \dfrac{V_{T, L}(v, \rho_{\thr}(u_{\cut}))}{\widehat{V}_{T, L}(v, \rho_{\thr}(u_{\cut}))}
\\
& \quad \leq \alpha \cdot \dfrac{V_{T, L}(v, \rho_{\thr}(u_{\cut}))}{\widehat{V}_{T, L}^{\prime}(v, \rho_{\thr}(u_{\cut}))}
\\
& \quad \eqcolon \alpha \cdot H_{T, L}(v, \rho_{\thr}(u_{\cut})),
\label{eq: proof - Theorem - Dependency-aware FDR control - 1}
\end{align}
where the inequality in the fourth line follows from the condition in~\eqref{eq: stopping time v} that all considered quadruples $(T, L, \rho_{\thr}(u_{\cut}), v)$ must satisfy. Taking the expectation of the FDP, we obtain
\begin{equation}
\FDR(v, \rho_{\thr}(u_{\cut}), T, L) \leq \alpha \cdot \mathbb{E} \big[ H_{T, L}(v, \rho_{\thr}(u_{\cut})) \big]
\label{eq: proof - Theorem - Dependency-aware FDR control - 2}
\end{equation}
and, thus, it remains to prove that $\mathbb{E} [ H_{T, L}(v, \rho_{\thr}(u_{\cut})) ] \leq 1$. Since~\eqref{eq: stopping time v} is a stopping time that is adapted to the filtration $\mathcal{F}_{v}$ in Lemma~\ref{Lemma: T-Rex+DA super-martingale} (i.e., $v$ is $\mathcal{F}_{v}$-measurable), we can apply Doob's optional stopping theorem to obtain an upper bound for $\mathbb{E} [ H_{T, L}(v, \rho_{\thr}(u_{\cut})) ]$, i.e.,
\begin{equation}
\mathbb{E} \big[ H_{T, L}(v, \rho_{\thr}(u_{\cut})) \big] \leq \mathbb{E} \big[ H_{T, L}(0.5, \rho_{\thr}(u_{\cut})) \big].
\label{eq: proof - Theorem - Dependency-aware FDR control - 3}
\end{equation}

Defining
\begin{align}
& \Psi_{t, L}^{+}(j, \rho_{\thr}(u_{\cut}))
\label{eq: dependency factor plus}
\\
& \coloneqq
\begin{cases}
 \dfrac{1}{2 \,\,\, - \,\,\, \smashoperator{\min\limits_{
 \substack{
 j^{\prime} \in \\
 \Gr(j, \rho_{\thr}(u_{\cut}))
 }
 }}  \big\lbrace \big| \Phi_{T, L}(j) - \Phi_{T, L}(j^{\prime}) \big| \big\rbrace}, 
&\!\!
\begin{array}{@{}c@{}}
\normalsize \Gr(j, \rho_{\thr}(u_{\cut})) \\
\neq \varnothing
\end{array}
 \\[2em]
 1, 
 &\!\!
 \begin{array}{@{}c@{}}
\normalsize \Gr(j, \rho_{\thr}(u_{\cut})) \\
= \varnothing
\end{array}
 \end{cases}.
\\
& \geq
\begin{cases}
 \dfrac{1}{2 \,\,\, - \,\,\, \smashoperator{\min\limits_{
 \substack{
 j^{\prime} \in \\
 \Gr(j, \rho_{\thr}(u_{\cut}))
 }
 }}  \big\lbrace \big| \Phi_{T, L}(j) - \Phi_{T, L}(j^{\prime}) \big| \big\rbrace}, 
&\!\!
\begin{array}{@{}c@{}}
\normalsize \Gr(j, \rho_{\thr}(u_{\cut})) \\
\neq \varnothing
\end{array}
 \\[2em]
 1 / 2, 
 &\!\!
 \begin{array}{@{}c@{}}
\normalsize \Gr(j, \rho_{\thr}(u_{\cut})) \\
= \varnothing
\end{array}
 \end{cases}.
 \\
 & =  \Psi_{t, L}(j, \rho_{\thr}(u_{\cut})),
\end{align}
$(t, L) \in \lbrace 1, \ldots, T \rbrace \times \mathbb{N}_{+}$, $H_{T, L}(0.5, \rho_{\thr}(u_{\cut}))$ can be upper bounded as follows:
\begin{align}
H_{T, L}(0.5, \rho_{\thr}(u_{\cut})) &= \dfrac{V_{T, L}(0.5, \rho_{\thr}(u_{\cut}))}{\widehat{V}_{T, L}^{\prime}(0.5, \rho_{\thr}(u_{\cut}))}
\\
&\leq \dfrac{V_{T, L}^{+}(0.5, \rho_{\thr}(u_{\cut}))}{\widehat{V}_{T, L}^{\prime \, +}(0.5, \rho_{\thr}(u_{\cut}))}
\\
&\eqqcolon H_{T, L}^{+}(0.5, \rho_{\thr}(u_{\cut})).
\label{eq: upper bound on H_T_L(0.5, rho_thr)}
\end{align}
The inequality in the second line follows from
\begin{align}
\text{(i) } V_{T, L}^{+}&(v, \rho_{\thr}(u_{\cut}))
\label{eq: upper bound on V_T_L(0.5, rho_thr)}
\\
&\coloneqq \lbrace \text{null } j : \Psi_{T, L}^{+}(j, \rho_{\thr}(u_{\cut})) \cdot \Phi_{T, L}(j) > v \rbrace
\\
&\geq \lbrace \text{null } j : \Psi_{T, L}(j, \rho_{\thr}(u_{\cut})) \cdot \Phi_{T, L}(j) > v \rbrace
\\
& = V_{T, L}(v, \rho_{\thr}(u_{\cut})),
\\[1em]
\text{(ii) } \widehat{V}_{T, L}^{\prime \, +}&(0.5, \rho_{\thr}(u_{\cut}))
\label{eq: lower bound on V_hat_prime_T_L(0.5, rho_thr)}
\\
&\coloneqq \sum\limits_{t = 1}^{T} \dfrac{p - \sum\limits_{q = 1}^{p}\Psi_{t, L}^{+}(q, \rho_{\thr}(u_{\cut})) \cdot \Phi_{t, L}(q)}{L - (t - 1)}
\\
& \leq \sum\limits_{t = 1}^{T} \dfrac{p - \sum\limits_{q = 1}^{p}\Psi_{t, L}(q, \rho_{\thr}(u_{\cut})) \cdot \Phi_{t, L}(q)}{L - (t - 1)}
\\
&= \widehat{V}_{T, L}^{\prime}(0.5, \rho_{\thr}(u_{\cut})).
\end{align}

Next, we show that $H_{T, L}^{+}(0.5, \rho_{\thr}(u_{\cut}))$ is monotonically increasing in $\rho_{\thr}(u_{\cut})$, i.e.,
\begin{align}
H_{T, L}^{+}(0.5, \rho_{\thr}(u_{\cut} + 1)) &= \dfrac{V_{T, L}^{+}(0.5, \rho_{\thr}(u_{\cut} + 1))}{\widehat{V}_{T, L}^{\prime \, +}(0.5, \rho_{\thr}(u_{\cut} + 1))}
\\
& \geq \dfrac{V_{T, L}^{+}(0.5, \rho_{\thr}(u_{\cut}))}{\widehat{V}_{T, L}^{\prime \, +}(0.5, \rho_{\thr}(u_{\cut}))}
\\
& = H_{T, L}^{+}(0.5, \rho_{\thr}(u_{\cut})),
\label{eq: proof - Theorem - Dependency-aware FDR control - 4}
\end{align}
where the inequality in the second line follows from Lemmas~\ref{Lemma: V_T,L monotonically increasing in rho} and~\ref{Lemma: V'_hat_T,L monotonically decreasing in rho}. Combining these preliminaries and noting that $\Psi_{t, L}^{+}(j, \rho_{\thr}(u_{\cut})) = 1$ yields
\begin{align}
V_{T, L}(0.5) &\coloneqq \lbrace \text{null } j : \Phi_{T, L}(j) > 0.5 \rbrace,
\label{eq: ordinary V_T_L}
\\[1em]
\widehat{V}_{T, L}^{\prime}(0.5) &\coloneqq \sum\limits_{t = 1}^{T} \dfrac{p - \sum\limits_{q = 1}^{p} \Phi_{t, L}(q)}{L - (t - 1)}
\label{eq: ordinary V_hat_prime_T_L}
\\[1em]
H_{T, L}(0.5) &\coloneqq \dfrac{V_{T, L}(0.5)}{\widehat{V}_{T, L}^{\prime}(0.5)}
\label{eq: ordinary H_T_L}
\end{align}
i.e., the counterparts of the dependency-aware expressions $V_{T, L}^{+}(0.5, \rho_{\thr}(u_{\cut}))$, $V_{T, L}^{\prime \, +}(0.5, \rho_{\thr}(u_{\cut}))$, and $H_{T, L}^{+}(0.5, \rho_{\thr}(u_{\cut}))$, respectively, we finally obtain
\begin{align}
\mathbb{E} \big[ H_{T, L}(v, \rho_{\thr}(u_{\cut})) \big] & \leq \mathbb{E} \big[ H_{T, L}(0.5, \rho_{\thr}(u_{\cut})) \big]
\\
& \leq \mathbb{E} \big[ H_{T, L}^{+}(0.5, \rho_{\thr}(u_{\cut})) \big]
\\
& \leq \mathbb{E} \big[ H_{T, L}^{+}(0.5, \rho_{\thr}(p)) \big]
\\
& = \mathbb{E} \big[ H_{T, L}^{+}(0.5, 1) \big]
\\
& = \mathbb{E} \big[ H_{T, L}(0.5) \big]
\\
& \leq 1,
\label{eq: proof - Theorem - Dependency-aware FDR control - 5}
\end{align}
where the first inequality is the same as in~\eqref{eq: proof - Theorem - Dependency-aware FDR control - 3}, the second and third inequalities follow from~\eqref{eq: upper bound on H_T_L(0.5, rho_thr)} and~\eqref{eq: proof - Theorem - Dependency-aware FDR control - 4}, respectively, and the equation in the fourth line follows from~\eqref{eq: cutoff dendrogram} (i.e., $\rho_{\thr}(p) = \sum_{u = 1}^{p} \Delta\rho_{\thr, u} = 1$). The equation in the fifth line is a consequence of the following: For $\rho_{\thr}(p) = 1$, we have $\Gr(j, 1) = \varnothing$ for all $j \in \lbrace 1, \ldots, p \rbrace$, and, therefore, it holds that $\Psi_{t, L}^{+}(j, 1) = 1$ for all $j \in \lbrace 1, \ldots, p \rbrace$. Thus, $H_{T, L}^{+}(0.5, 1)$ boils down to its ordinary counterpart $H_{T, L}(0.5)$ in~\eqref{eq: ordinary H_T_L}, i.e., $H_{T, L}^{+}(0.5, 1) = H_{T, L}(0.5)$. Finally, the proof of Inequality~\eqref{eq: proof - Theorem - Dependency-aware FDR control - 5} is omitted because it is exactly as in the proof of Theorem~1 (FDR control) in~\cite{machkour2021terminating}.
\label{proof: Theorem - Dependency-aware FDR control}
\end{proof}

\subsection{General Group Design Principle}
\label{subsec: General Group Design Principle}
The \textit{T-Rex+DA} selector is not only suitable for the considered binary tree graphs or dendrograms but also for various other dependency models. In fact, a closer look at Lemmas~\ref{Lemma: V_T,L monotonically increasing in rho} and~\ref{Lemma: V'_hat_T,L monotonically decreasing in rho}, which are essential for the proof of Theorem~\ref{Theorem: Dependency-aware FDR control} (dependency-aware FDR control), reveals that the following general design principle for the groups $\Gr(j, \rho_{\thr})$ can be derived from Lemmas~\ref{Lemma: V_T,L monotonically increasing in rho} and~\ref{Lemma: V'_hat_T,L monotonically decreasing in rho}:
\begin{thm}[Group design principle]
Consider the generic definition of the variable groups in Definition~\ref{definition: Dependency-aware relative occurrences}, i.e., $\Gr(j, \rho_{\thr}) \in \lbrace 1, \ldots, p \rbrace \backslash \lbrace j \rbrace$, $j = 1, \ldots, p$, $\rho_{\thr} \in [0, 1]$. If any $\rho_{1}$, $\rho_{2} \in [0, 1]$, $\rho_{2} > \rho_{1}$, satisfy
\begin{equation}
\Gr(j, \rho_{2}) \subseteq \Gr(j, \rho_{1}), \quad j = 1, \ldots, p,
\label{eq: Group design principle}
\end{equation}
then the \textit{T-Rex+DA} selector controls the FDR at the target level $\alpha \in [0, 1]$.
\label{Theorem: Group design principle}
\end{thm}
\begin{proof}
The FDR control property in Theorem~\ref{Theorem: Dependency-aware FDR control} holds if Lemmas~\ref{Lemma: V_T,L monotonically increasing in rho} and~\ref{Lemma: V'_hat_T,L monotonically decreasing in rho} hold. Lemmas~\ref{Lemma: V_T,L monotonically increasing in rho} and~\ref{Lemma: V'_hat_T,L monotonically decreasing in rho} hold for any definition of $\Gr(j, \rho_{\thr})$ that satisfies the group design principle in Theorem~\ref{eq: Group design principle}.
\label{proof: Lemma - Group design principle}
\end{proof}

Loosely speaking, Theorem~\ref{Theorem: Group design principle} states that the cardinalities of any variable group $j$ must be monotonically decreasing in $\rho_{\thr}$ and follow the subset structure illustrated in Figure~\ref{fig: Illustration of group design principle}. Thus, any dependency model (e.g., graph models, time series models, equicorrelated models, etc.) that follows the design principle in Theorem~\ref{Theorem: Group design principle} can be incorporated into the \textit{T-Rex+DA} selector. This property makes the \textit{T-Rex+DA} selector a versatile method that can cope with various dependency models.
%
\begin{figure}
\centering
\begin{tikzpicture}[scale=0.9,transform shape]
\draw [thick, fill=none] (0,0) arc (-90:270:1.3cm and 0.75cm);
\draw [thick, fill=none] (0.25,-0.2) arc (-90:270:2.3cm and 1.25cm);
\draw [thick, fill=none] (0.5,-0.4) arc (-90:270:3.15cm and 1.7cm);
\draw [thick, fill=none] (0.75,-0.6) arc (-90:270:4.1cm and 2.35cm);
\node [xshift=0cm, yshift=0.7cm] (0,0) {\small $\Gr(j, \rho_{\thr} = 1)$};
\node [xshift=1.25cm, yshift=1.65cm] (0,0) {\small $\Gr(j, \rho_{\thr})$};
\node [xshift=2cm, yshift=2.45cm] (0,0) {\small $\iddots$};
\node [xshift=2.6cm, yshift=3.15cm] (0,0) {\small $\Gr(j, \rho_{\thr} = 0)$};
\end{tikzpicture}
\caption{Illustration of the group design principle for dependency-aware FDR control in Theorem~\ref{Theorem: Group design principle}.}
\label{fig: Illustration of group design principle}
\end{figure}
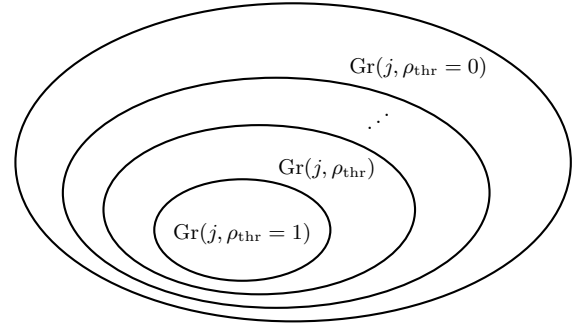
%

\section{Optimal Dependency-Aware T-Rex Algorithm}
\label{sec: Optimal Dependency-Aware T-Rex Algorithm}
In this section, we propose an efficient calibration algorithm and prove that it yields an optimal solution of~\eqref{eq: optimization problem - FDR-controlled variable selection}. That is, it optimally calibrates the parameters $v$, $\rho_{\thr}(u_{\cut})$, and $T$ of the proposed \textit{T-Rex+DA} selector, such that the FDR is controlled at the target level while maximizing the number of selected variables.The pseudocode of the proposed \textit{T-Rex+DA} calibration algorithm is given in Algorithm~\ref{algorithm: Extended T-Rex+DA}. An open source implementation of the proposed calibration algorithm for the \textit{T-Rex+DA} selector is available within the R package `TRexSelector' on CRAN~\cite{machkour2022TRexSelector}.
\begin{algorithm}[h]
\caption{Extended \textit{T-Rex+DA} Calibration.}
\begin{alglist}
\item \textbf{Input}: $\alpha \in [0, 1]$, $\X$, $\y$, $K$, $\tilde{v}$, $\tilde{\rho}_{\thr}$, $L_{\max}$, $T_{\max}$.
\label{algorithm: Extended T-Rex+DA Step 1}
\item \textbf{Set} $L = p$, $T = 1$.
\label{algorithm: Extended T-Rex+DA Step 2}
\item \textbf{While} $\widehat{\FDP}(v = \tilde{v}, \rho_{\thr}(u_{\cut}) = \tilde{\rho}_{\thr}, T, L) > \alpha$ and\\
$L \leq L_{\max}$ \textbf{do}:
\label{algorithm: Extended T-Rex+DA Step 3}
\begin{enumerate}
\item[] \textbf{Set} $L \leftarrow L + p$.
\end{enumerate}
\item \textbf{Set} $\Delta v \!=\! \dfrac{1}{K}$, $\widehat{\FDP}(v = 1 - \Delta v, \rho_{\thr}(u_{\cut}) = \tilde{\rho}_{\thr}, T, L) = 0$.
\label{algorithm: Extended T-Rex+DA Step 4}
\item \textbf{While} $\widehat{\FDP}(v = 1 - \Delta v, \rho_{\thr}(u_{\cut}) = \tilde{\rho}_{\thr}, T, L) \leq \alpha$ and\\
$T \leq T_{\max}$ \textbf{do}:
\begin{enumerate}[label*=\arabic*.]
\item[5.1.] \textbf{For} $v = 0.5, 0.5 + \Delta v, 0.5 + 2 \cdot \Delta v, \ldots, 1 - \Delta v$ \textbf{do}:
\begin{enumerate}[label*=\arabic*.]
\item[5.1.1.] \textbf{For} $u_{\cut} = 1, \ldots, p$ \textbf{do}:
\begin{enumerate}
\item[i.] \textbf{Compute} $\widehat{\FDP}(v, \rho_{\thr}(u_{\cut}), T, L)$ as in Def.~\ref{definition: dependency-aware FDP estimator}.
\item[ii.] \textbf{If} $\widehat{\FDP}(v, \rho_{\thr}(u_{\cut}), T, L) \leq \alpha$
\begin{enumerate}
\item[] \textbf{Compute} $\widehat{\A}_{L}(v, \rho_{\thr}(u_{\cut}), T)$ as in~\eqref{eq: selected active set T-Rex+DA}.
\end{enumerate}
\textbf{Else}
\begin{enumerate}
\item[] \textbf{Set} $\widehat{\A}_{L}(v, \rho_{\thr}(u_{\cut}), T) = \varnothing$.
\end{enumerate}
\end{enumerate}
\end{enumerate}
\item[5.2.] \textbf{Set} $T \leftarrow T + 1$.
\end{enumerate}
\label{algorithm: Extended T-Rex+DA Step 5}
\item \textbf{Solve}
\begin{alignat*}{2}
& \max_{v^{\prime}, \rho_{\thr}(u_{\cut}^{\prime}), T^{\prime}} &{}& \, \big| \widehat{\A}_{L}(v^{\prime}, \rho_{\thr}(u_{\cut}^{\prime}), T^{\prime}) \big| \\
& \text{subject to} &{}& T^{\prime} \in \lbrace 1, \ldots, T - 1 \rbrace \\
&{}&{}& u_{\cut}^{\prime} \in \lbrace 1, \ldots, p \rbrace \\
&{}&{}& v^{\prime} \in \lbrace  0.5, 0.5 + \Delta v, 0.5 + 2 \Delta v, \ldots, 1 - \Delta v \rbrace
\end{alignat*}
and let $(v^{*}, \rho_{\thr}(u_{\cut}^{*}), T^{*})$ be a solution.
\label{algorithm: Extended T-Rex+DA Step 6}
\item \textbf{Output}: $(v^{*}, \rho_{\thr}(u_{\cut}^{*}), T^{*}, L)$ and $\widehat{\A}_{L}(v^{*}, \rho_{\thr}(u_{\cut}^{*}), T^{*})$.
\label{algorithm: Extended T-Rex+DA Step 7}
\end{alglist}
\label{algorithm: Extended T-Rex+DA}
\end{algorithm}

First, some hyperparameters, which are relevant for managing the tradeoff between achieving a high TPR, the memory consumption, and computation time but have no influence on the FDR control property of the proposed method, are set. Throughout this paper, the hyperparameters are chosen based on the suggestions in~\cite{machkour2021terminating} as follows: $\tilde{v} = 0.75$, $\tilde{\rho}_{\thr} = \rho_{\thr}(\lfloor 0.75 \cdot p \rceil)$, $L_{\max} = 10p$, $T_{\max} = \lceil n / 2 \rceil$, where $\lfloor \cdot \rceil$ and $\lceil \cdot \rceil$ denote rounding towards the nearest integer and the nearest higher integer, respectively. Moreover, it was shown for various applications and in extensive simulations that there are no improvements for $K > 20$ random experiments~\cite{machkour2021terminating,machkour2022TRexGVS,machkour2023ScreenTRex} and, therefore, we set $K = 20$.

The algorithm proceeds as follows: It takes the user-defined target FDR, the original predictor matrix $\X$, and the response vector $\y$ as inputs. Then, it determines $L$ via a loop that adds $p$ dummies in each iteration until $\widehat{\FDP}(v = \tilde{v}, \rho_{\thr}(u_{\cut}) = \tilde{\rho}_{\thr}, T = 1, L)$ falls below the target FDR level $\alpha$ at a reference point $(v, \rho_{\thr}(u_{\cut}), T) = (\tilde{v}, \tilde{\rho}_{\thr}, 1)$ or $L$ reaches the maximum allowed value $L_{\max}$. This guarantees that the FDR is controlled as tightly as possible at the target FDR level $\alpha$ while ensuring that the TPR is as high as possible. Supposing that there exists no $\tilde{\rho}_{\thr}^{\prime} \neq \tilde{\rho}_{\thr}$ that satisfies $\widehat{\FDP}(v = 1 - \Delta v, \rho_{\thr}(u_{\cut}) = \tilde{\rho}_{\thr}^{\prime}, T, L) < \widehat{\FDP}(v = 1 - \Delta v, \rho_{\thr}(u_{\cut}) = \tilde{\rho}_{\thr}, T, L)$, the algorithm proceeds by increasing the number of included dummies $T$ at a reference point $(v, \rho_{\thr}(u_{\cut})) = (1 - \Delta v, \tilde{\rho}_{\thr})$, where $\Delta v = 1 / K$, until $\widehat{\FDP}(v = 1 - \Delta v, \rho_{\thr}(u_{\cut}) = \tilde{\rho}_{\thr}, T, L)$ exceeds the target level $\alpha$ or reaches the maximum allowed number of included dummies $T_{\max}$. Fixing the optimized parameters $L$, $T$ and the corresponding FDR-controlled variable sets, the algorithm then determines the optimal values of $v$ and $\rho_{\thr}(u_{\cut})$ that maximize the number of selected variables by solving the optimization problem in~\eqref{algorithm: Extended T-Rex+DA Step 6}. Finally, the obtained solution $(v^{*}, \rho_{\thr}(u_{\cut}^{*}), T^{*}, L)$ yields the FDR-controlled set of selected variables
\begin{align}
\widehat{\A}_{L}(v^{*}, \rho_{\thr}(u_{\cut}^{*}), T^{*}) = \big\lbrace j : \Phi_{T^{*}, L}^{\DA}(j, \rho_{\thr}(u_{\cut}^{*})) > v^{*} \big\rbrace.\quad
\label{eq: selected active set T-Rex+DA}
\end{align}

In the following, we state and prove that Algorithm~\ref{algorithm: Extended T-Rex+DA} yields an optimal solution of~\eqref{eq: optimization problem - FDR-controlled variable selection}:

\begin{thm}[Optimal Dependency-Aware Calibration]
Suppose that $L$, as obtained by Algorithm~\ref{algorithm: Extended T-Rex+DA}, is fixed. Then, any triple $(v^{*}, \rho_{\thr}(u_{\cut}^{*}), T^{*})$ of a quadruple $(v^{*}, \rho_{\thr}(u_{\cut}^{*}), T^{*}, L)$, as obtained by Algorithm~\ref{algorithm: Extended T-Rex+DA}, is an optimal solution of~\eqref{eq: optimization problem - FDR-controlled variable selection}.
\label{Theorem: Optimal Dependency-Aware Calibration Algorithm}
\end{thm}
\begin{proof}
From Equation~\eqref{eq: number of selected variables}, it follows that for all quadruples $(v, \rho_{\thr}(u_{\cut}), T, L)$ that satisfy $\widehat{\FDP}(v, \rho_{\thr}(u_{\cut}), T, L) \leq \alpha$, the objective functions in Step~\ref{algorithm: Extended T-Rex+DA Step 6} of Algorithm~\ref{algorithm: Extended T-Rex+DA} and in the optimization problem in~\eqref{eq: optimization problem - FDR-controlled variable selection} are equal, i.e., $| \widehat{\A}_{L}(v^{\prime}, \rho_{\thr}(u_{\cut}^{\prime}), T^{\prime}) | = R_{T^{\prime}, L}(v^{\prime}, \rho_{\thr}(u_{\cut}^{\prime}))$. Therefore, and since all attainable values of $u_{\cut}$ are considered in Step~\ref{algorithm: Extended T-Rex+DA Step 6} of Algorithm~\ref{algorithm: Extended T-Rex+DA}, it suffices to show that for fixed $\rho_{\thr}(u_{\cut})$ and $L$, the set of feasible tuples $(v, T)$ of~\eqref{eq: optimization problem - FDR-controlled variable selection} is a subset of or equal to the set of feasible tuples obtained by Algorithm~\ref{algorithm: Extended T-Rex+DA}. Since, ceteris paribus, $\widehat{\FDP}(v, \rho_{\thr}(u_{\cut}), T, L)$ is monotonically decreasing in $v$, and monotonically increasing in $T$~\cite{machkour2021terminating}, for $T = T_{\fin}$, $\widehat{\FDP}(v, \rho_{\thr}(u_{\cut}), T, L)$ attains its minimum value at $(v, T) = (1 - \Delta v, T_{\fin})$, where $T_{\fin} \in \lbrace 1, \ldots, L \rbrace$ is implicitly defined through the inequalities
\begin{equation}
\widehat{\FDP}(1 - \Delta v, \rho_{\thr}(u_{\cut}), T_{\fin}, L) \leq \alpha
\label{eq: proof - Theorem - Optimal Dependency-Aware Calibration Algorithm 1}
\end{equation}
and
\begin{equation}
\widehat{\FDP}(1 - \Delta v, \rho_{\thr}(u_{\cut}), T_{\fin} + 1, L) > \alpha.
\label{eq: proof - Theorem - Optimal Dependency-Aware Calibration Algorithm 2}
\end{equation}
Thus, the feasible set of the optimization problem in~\eqref{eq: optimization problem - FDR-controlled variable selection} is given by
\begin{align}
\big\lbrace (v, & T) : \widehat{\FDP}(v, \rho_{\thr}(u_{\cut}), T, L) \leq \alpha \rbrace
\\
& \qquad= \lbrace  (v, T) : v \in [ 0.5, 1 - \Delta v ],
\\
& \qquad\qquad\qquad\quad\, T \in \lbrace 1, \ldots, T_{\fin} \rbrace,
\\
& \qquad\qquad\qquad\quad\, \widehat{\FDP}(v, \rho_{\thr}(u_{\cut}), T, L) \leq \alpha \big\rbrace.
\label{eq: proof - Theorem - Optimal Dependency-Aware Calibration Algorithm 3}
\end{align}
Since $\Delta v = 1 / K$, the upper endpoint of the interval $[0.5, 1 - \Delta v]$ asymptotically (i.e., $K \rightarrow \infty$) coincides with the supremum of the interval $[0.5, 1)$. That is, the set in~\eqref{eq: proof - Theorem - Optimal Dependency-Aware Calibration Algorithm 3} contains all feasible solutions of~\eqref{eq: optimization problem - FDR-controlled variable selection}. However, since the $v$-grid in Algorithm~\ref{algorithm: Extended T-Rex+DA} is, as in \cite{machkour2021terminating}, adapted to $K$, all values of $R_{T, L}(v, \rho_{\thr}(u_{\cut}))$ that are attained by off-grid solutions can also be attained by on-grid solutions. Thus, instead of~\eqref{eq: proof - Theorem - Optimal Dependency-Aware Calibration Algorithm 3} only the following fully discrete feasible set of~\eqref{eq: optimization problem - FDR-controlled variable selection} needs to be considered:
\begin{align}
\big\lbrace  (v, T)  : 
\,\,& v \in \lbrace 0.5, 0.5 + \Delta v, 0.5 + 2 \Delta v, \ldots, 1 - \Delta v \rbrace,
\\
& T \in \lbrace 1, \ldots, T_{\fin} \rbrace,
\\
& \widehat{\FDP}(v, \rho_{\thr}(u_{\cut}), T, L) \leq \alpha \big\rbrace.
\label{eq: proof - Theorem - Optimal Dependency-Aware Calibration Algorithm 4}
\end{align}
Since the ``while''-loop in Step~\ref{algorithm: Extended T-Rex+DA Step 5} of Algorithm~\ref{algorithm: Extended T-Rex+DA} is terminated when $T = T_{\fin} + 1$, the feasible set of the optimization problem in Step~\ref{algorithm: Extended T-Rex+DA Step 6} of Algorithm~\ref{algorithm: Extended T-Rex+DA} is given by
\begin{align}
\big\lbrace  (v, T) : 
\,\,& v \in \lbrace 0.5, 0.5 + \Delta v, 0.5 + 2 \Delta v, \ldots, 1 - \Delta v \rbrace,
\\
& T \in \lbrace 1, \ldots, T_{\fin} \rbrace
\\
& \widehat{\FDP}(v, \rho_{\thr}(u_{\cut}), T, L) \leq \alpha \big\rbrace,
\label{eq: proof - Theorem - Optimal Dependency-Aware Calibration Algorithm 5}
\end{align}
which is equal to~\eqref{eq: proof - Theorem - Optimal Dependency-Aware Calibration Algorithm 4}.
\label{proof: Theorem - Optimal Dependency-Aware Calibration Algorithm}
\end{proof}

\section{Numerical Experiments}
\label{sec: Numerical Experiments}
In this section, we verify the FDR control property of the proposed \textit{T-Rex+DA} selector via numerical experiments and compare its performance against three state-of-the-art methods for high-dimensional data, i.e., \textit{model-X} knockoff~\cite{candes2018panning}, \textit{model-X} knockoff+~\cite{candes2018panning}, and \textit{T-Rex} selector~\cite{machkour2021terminating}.
%
\begin{figure*}[!htbp]
  \centering
  \subfloat[]{
  		\scalebox{1}{
  			\includegraphics[width=0.37\linewidth]{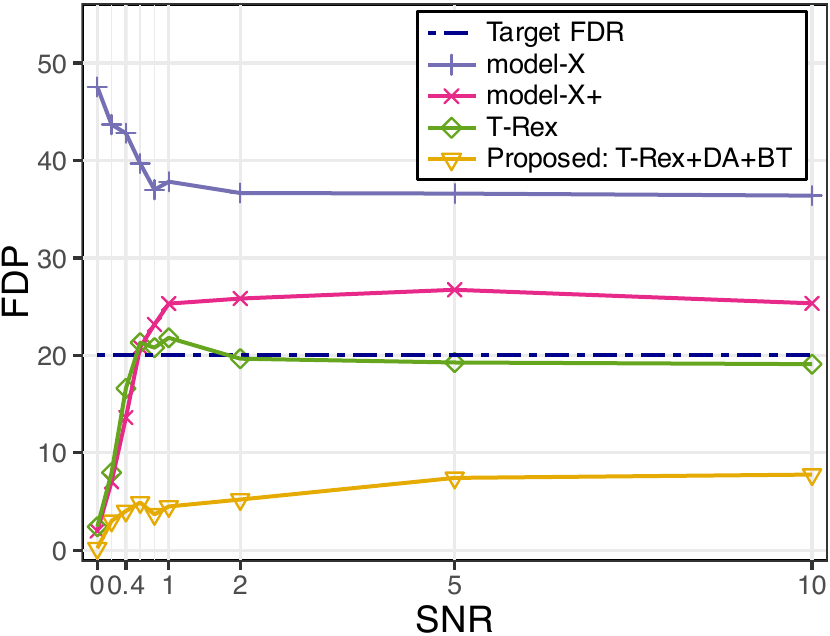}
  		}
   		\label{fig: FDP_vs_SNR_tFDR_20_n_150_rho_07}
   }
	\hspace*{2em}
  \subfloat[]{
  		\scalebox{1}{
  			\includegraphics[width=0.3765\linewidth]{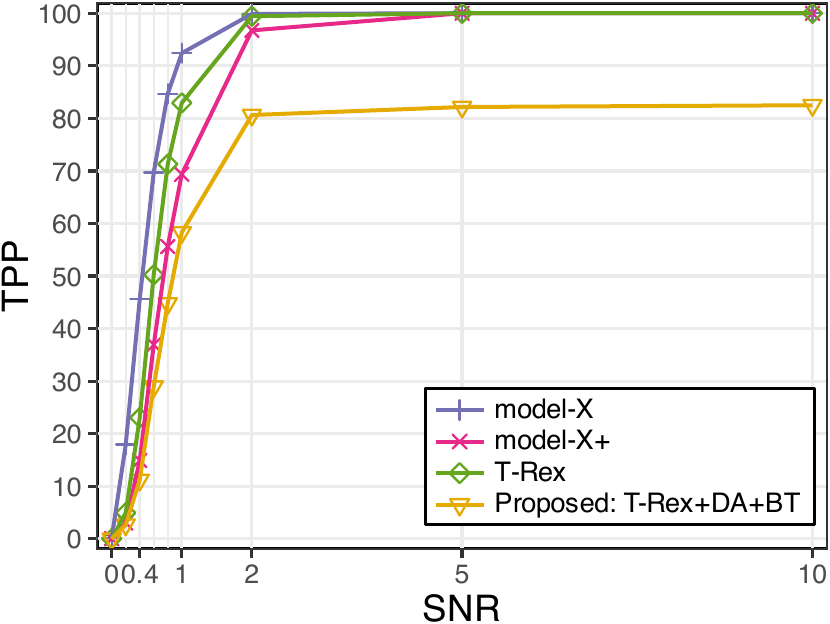}
  		}
   		\label{fig: TPP_vs_SNR_tFDR_20_n_150_rho_07}
   }
   \\
\vspace{-0.8em}
  \subfloat[]{
  		\scalebox{1}{
  			\includegraphics[width=0.37\linewidth]{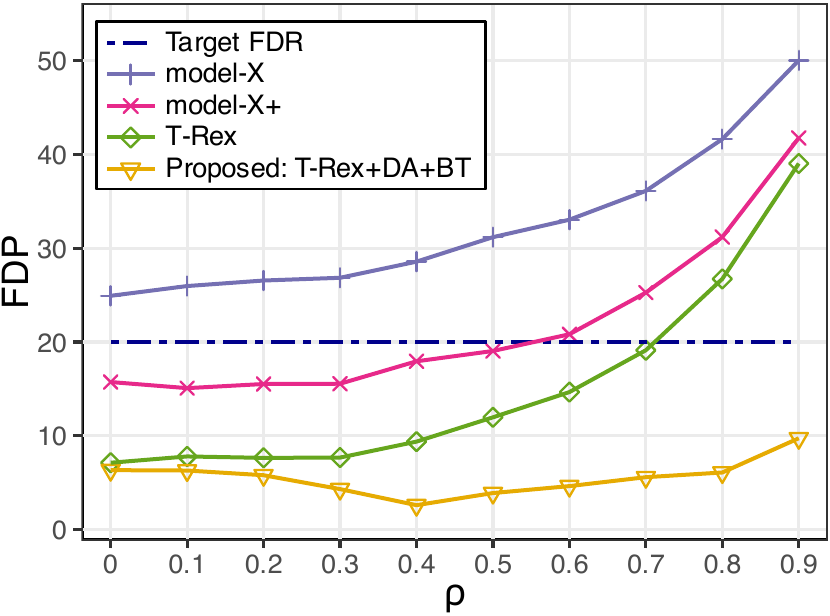}
  		}
   		\label{fig: FDP_vs_rho_tFDR_20_n_150_rho_07}
   }
	\hspace*{2em}
  \subfloat[]{
  		\scalebox{1}{
  			\includegraphics[width=0.3765\linewidth]{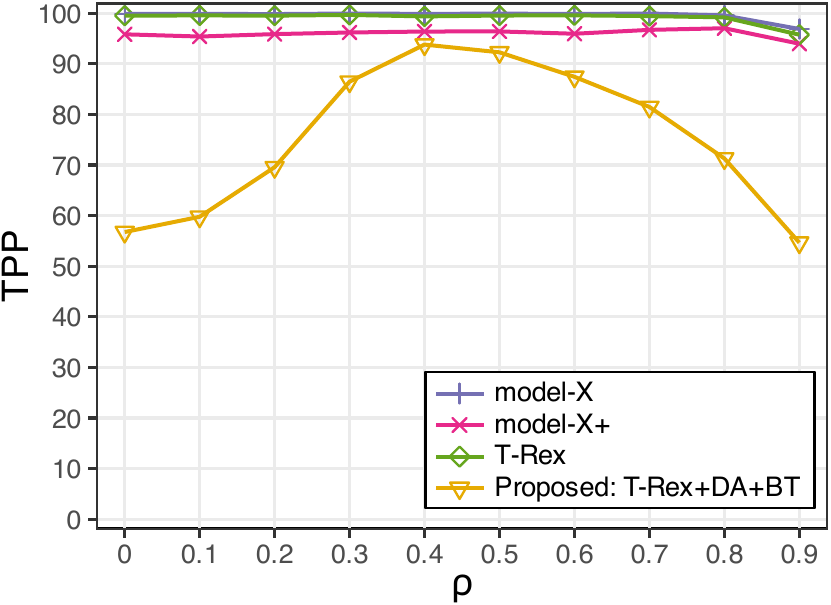}
  		}
   		\label{fig: TPP_vs_rho_tFDR_20_n_150_rho_07}
   }
   \\
\vspace{-0.8em}
  \subfloat[]{
  		\scalebox{1}{
  			\includegraphics[width=0.37\linewidth]{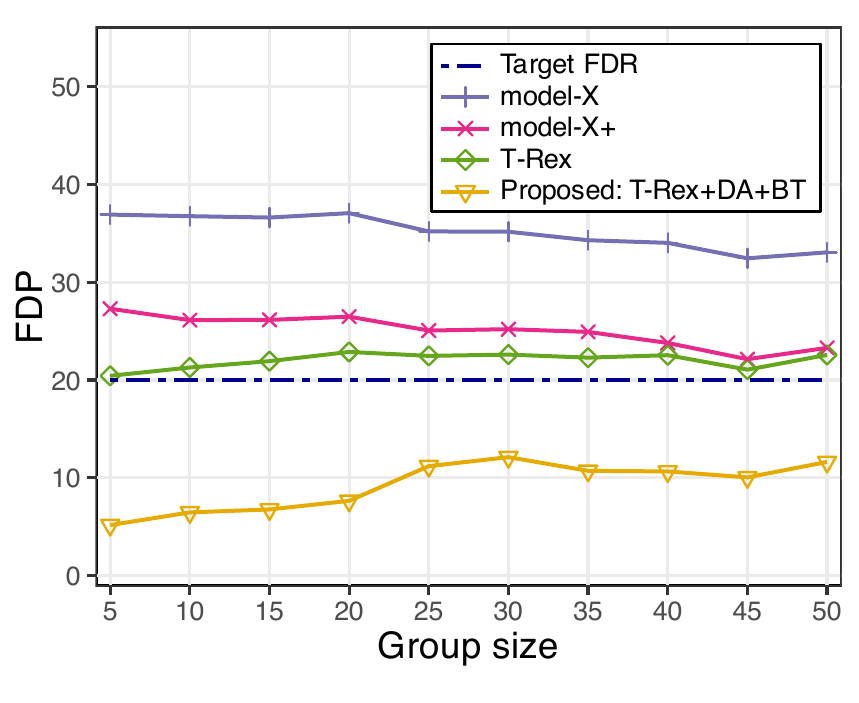}
  		}
   		\label{fig: FDP_vs_groupSize_tFDR_20_n_150_rho_07}
   }
	\hspace*{2em}
  \subfloat[]{
  		\scalebox{1}{
  			\includegraphics[width=0.3765\linewidth]{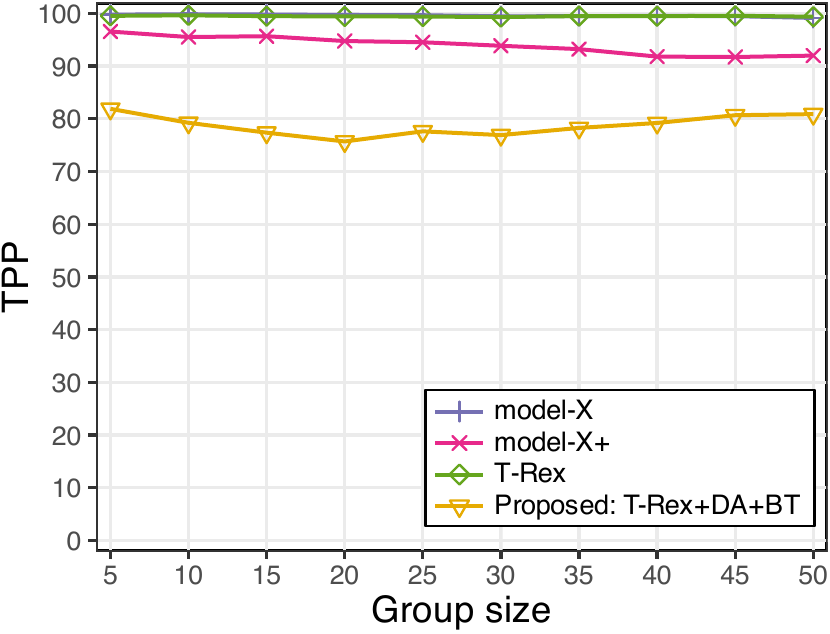}
  		}
   		\label{fig: TPP_vs_groupSize_tFDR_20_n_150_rho_07}
   }
   \\
\vspace{-0.8em}
  \subfloat[]{
  		\scalebox{1}{
  			\includegraphics[width=0.37\linewidth]{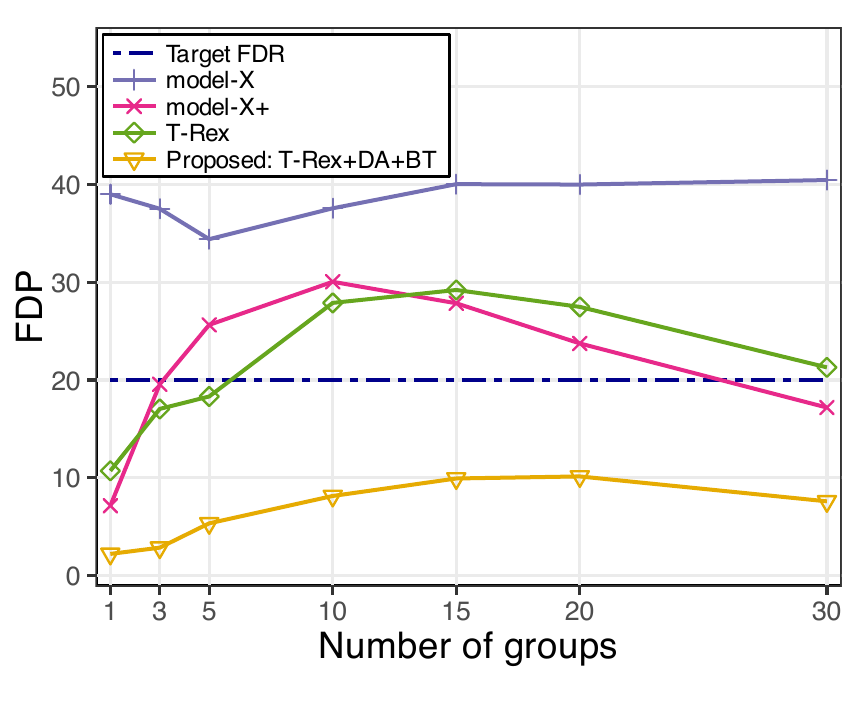}
  		}
   		\label{fig: FDP_vs_numGroups_tFDR_20_n_150_rho_07}
   }
	\hspace*{2em}
  \subfloat[]{
  		\scalebox{1}{
  			\includegraphics[width=0.3765\linewidth]{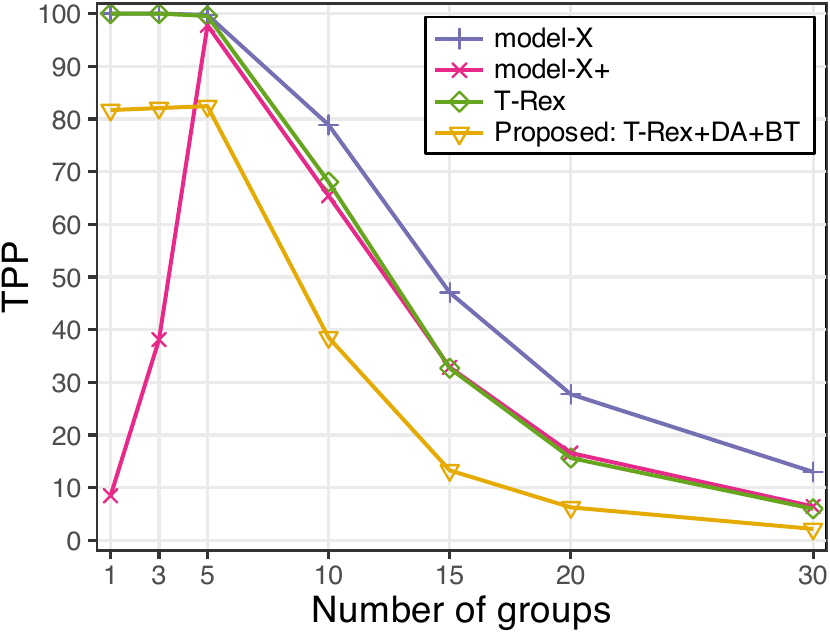}
  		}
   		\label{fig: TPP_vs_numGroups_tFDR_20_n_150_rho_07}
   }
  \caption{Only the proposed \textit{T-Rex+DA} selector with a binary tree group model (\textit{T-Rex+DA+BT}) reliably controls the FDR in all settings while achieving a reasonably high TPR. In Figure~(c), we see that with increasing correlations among the variables in a group, the benchmark methods exhibit an alarming increase in FDR.}
  \label{fig: sweep_plots_3_tFDR_20_n_150_rho_07}
\end{figure*}
%
\begin{figure}[!htbp]
  \centering
  \hspace*{-0.75em}
  \subfloat[$\rho = 0.7$]{
  		\scalebox{1}{
  			\includegraphics[width=0.485\linewidth]{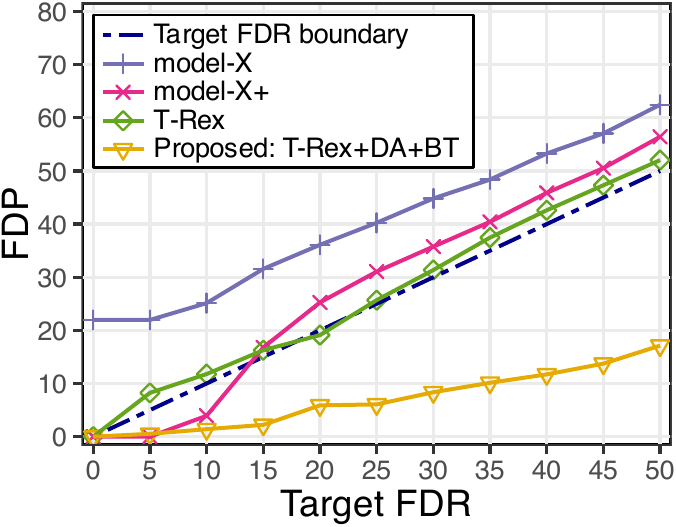}
  		}
   		\label{fig: FDP_vs_tFDR_n_150_p_500_rho_07}
   }
	\hspace*{-1.3em}
  \subfloat[$\rho = 0.7$]{
  		\scalebox{1}{
  			\includegraphics[width=0.485\linewidth]{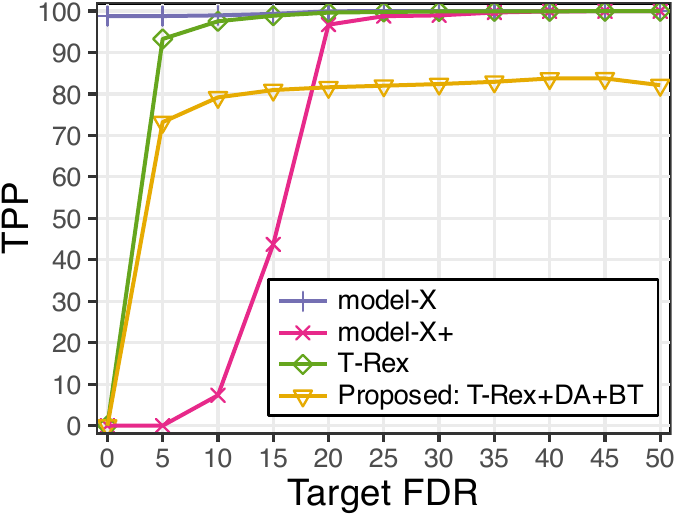}
  		}
   		\label{fig: TPP_vs_tFDR_n_150_p_500_rho_07}
   }
   \\
 \hspace*{-0.75em}
  \subfloat[$\rho = 0.8$]{
  		\scalebox{1}{
  			\includegraphics[width=0.485\linewidth]{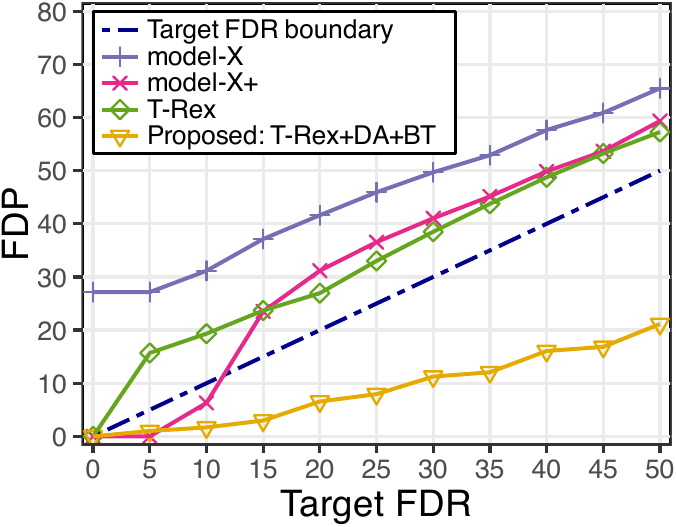}
  		}
   		\label{fig: FDP_vs_tFDR_n_150_p_500_rho_08}
   }
	\hspace*{-1.3em}
  \subfloat[$\rho = 0.8$]{
  		\scalebox{1}{
  			\includegraphics[width=0.485\linewidth]{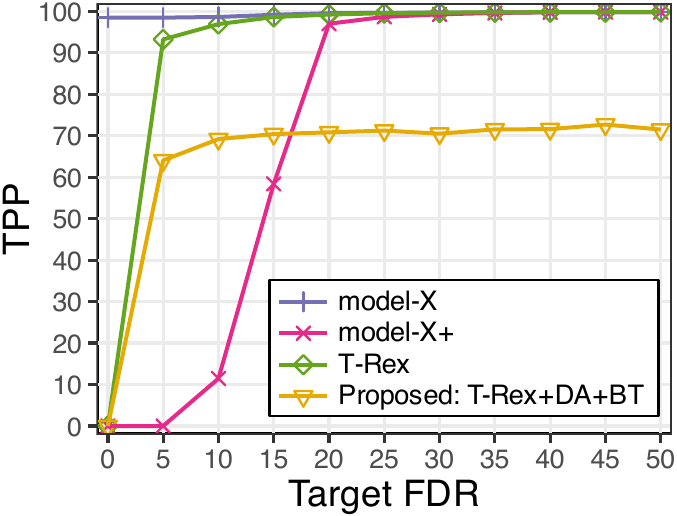}
  		}
   		\label{fig: TPP_vs_tFDR_n_150_p_500_rho_08}
   }
   \\
 \hspace*{-0.75em}
  \subfloat[$\rho = 0.9$]{
  		\scalebox{1}{
  			\includegraphics[width=0.485\linewidth]{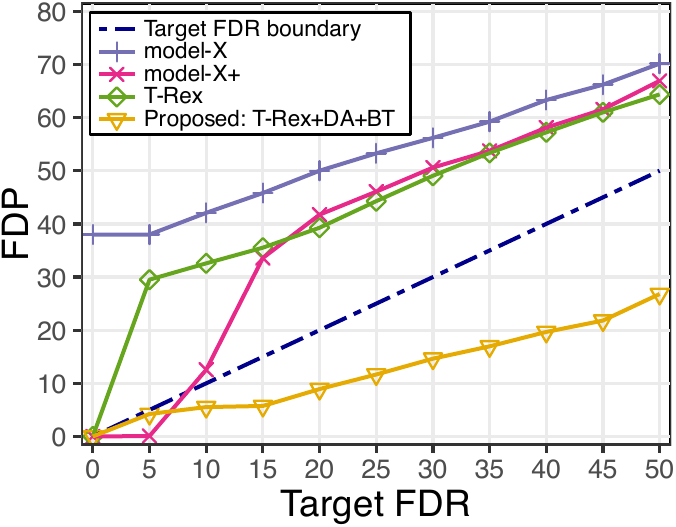}
  		}
   		\label{fig: FDP_vs_tFDR_n_150_p_500_rho_09}
   }
	\hspace*{-1.3em}
  \subfloat[$\rho = 0.9$]{
  		\scalebox{1}{
  			\includegraphics[width=0.485\linewidth]{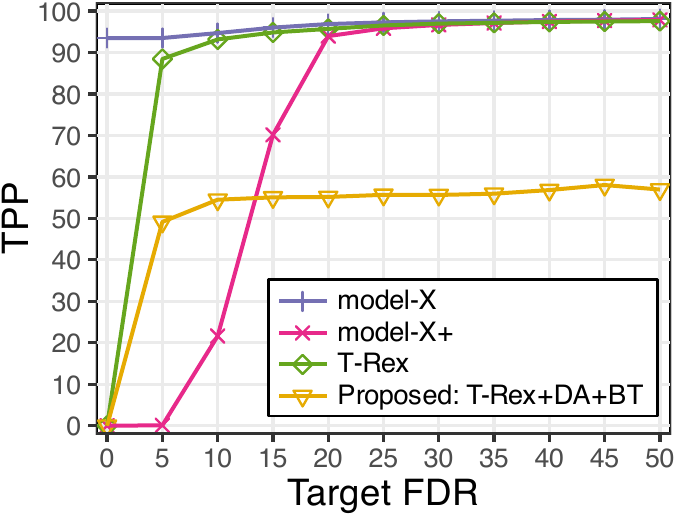}
  		}
   		\label{fig: TPP_vs_tFDR_n_150_p_500_rho_09}
   }
  \caption{The proposed \textit{T-Rex+DA} selector reliably controls the FDR in all settings while achieving a reasonably high TPR in harsh high correlation settings. We observe that with increasing correlations among the variables in a group, the benchmark methods do not control the FDR for almost any choice of target FDR.}
  \label{fig: sweep_plots_tFDR}
\end{figure}
%

We consider a high-dimensional setting with $p = 500$ variables and $n = 150$ samples and generate the predictor matrix $\X$ from a zero mean multivariate Gaussian distribution with an $M$ block diagonal correlation matrix, where each block is a $Q \times Q$ toeplitz correlation matrix, i.e., 
\begin{equation}
\begingroup
\renewcommand*{\arraystretch}{1}
\setlength\arraycolsep{1.5pt}
\bSigma \! = \!
\begin{bmatrix}
\bSigma_{1} & \boldsymbol{0} & \hdots & \boldsymbol{0}
\\
\boldsymbol{0} & \ddots &  & \vdots
\\
\vdots &  & \bSigma_{M} & \boldsymbol{0}
\\
\boldsymbol{0} & \hdots & \boldsymbol{0} & \boldsymbol{0}
\end{bmatrix}
\!\!,
\bSigma_{m} \! = \!
\begin{bmatrix}
1 & \rho & \rho^{2} & \cdots & \rho^{Q - 1}
\\
\rho & 1 & \rho & \cdots & \rho^{Q - 2}
\\
\rho^{2} & \rho & 1 & \cdots & \rho^{Q - 3}
\\
\vdots & \vdots & \vdots & \ddots & \vdots
\\
\rho^{Q - 1} & \rho^{Q - 2} & \rho^{Q - 3} & \hdots & 1
\end{bmatrix}.
\endgroup
\label{eq: block diagonal correlation matrix}
\end{equation}
That is, each block mimics a dependency structure that is often present in biomedical data (e.g., gene expression~\cite{segal2003regression} and genomics data~\cite{balding2006tutorial}) and may lead to the breakdown of the FDR control property of existing methods.
The response vector $\y$ is generated from the linear model $\y = \X\bbeta + \bepsilon$, where $\bbeta = [ \beta_{1} \, \cdots \, \beta_{p} ]^{\top} \in \mathbb{R}^{p}$ is the sparse true coefficient vector and $\epsilon \sim \mathcal{N}(\boldsymbol{0}, \sigma^{2} \boldsymbol{I})$ is an additive noise vector with variance $\sigma^{2}$ and identity matrix $\boldsymbol{I}$. The variance $\sigma^{2}$ is set such that the signal-to-noise ratio $\SNR = \Var [\X\bbeta] / \sigma^{2}$ has the desired value. In the base setting, we set the parameters as follows: $\SNR = 2$, $\rho = 0.7$, $Q = 5$, $M = 5$, $\alpha = 0.2$. The coefficient vector $\bbeta$ is generated such that the $m$th block consists of one true active variable with coefficient value one, while the remaining variables are nulls with coefficient value zero. In the numerical experiments, all parameters except for one parameter of the base setting are varied. That is, ceteris paribus, $\SNR$, $\rho$, group size $Q$, number of groups $M$, and target FDR $\alpha$ are varied. 

The results in Figures~\ref{fig: sweep_plots_3_tFDR_20_n_150_rho_07} and~\ref{fig: sweep_plots_tFDR} are averaged over $955$ Monte Carlo replications.\footnote{The uneven number of Monte Carlo replications was chosen to run the simulations efficiently and in parallel on the Lichtenberg High-Performance Computer of the Technische Universität Darmstadt.} For the performance comparison, we consider the averaged FDP and TPP (in \%) which are estimates of the FDR and TPR. We observe that only the proposed \textit{T-Rex+DA} selector with a binary tree group model (\textit{T-Rex+DA+BT}) reliably controls the FDR over all values of $\SNR$, $\rho$, $Q$, $M$, and $\alpha$, while the benchmarks lose the FDR control property, especially in the practically important case where groups of highly correlated variables are present in the data. It is remarkable that the frequently used \textit{model-X} knockoff method exceeds the target FDR level in all scenarios by far. Note that achieving a higher TPR without controlling the FDR is undesirable, since it leads to reporting false discoveries, which need to be avoided in order to alleviate the unfortunately still ongoing reproducibility crisis in many scientific fields~\cite{baker2016reproducibility}.

The results of two additional simulation setups, where 
\begin{enumerate}
\item the $p$-dimensional samples of the predictor matrix (i.e., rows of $\X$) are sampled from a zero-mean multivariate heavy-tailed Student-$t$ distribution with covariance matrix $\bSigma$ and $3$ degrees of freedom,
\item the noise vector $\bepsilon$ is sampled from a heavy-tailed Student-$t$ distribution with $3$ degrees of freedom,
\end{enumerate}
are deferred to Appendix~\ref{appendix: Additional Simulations} in the supplementary materials. These additional simulations verify the theoretical results and show that only the proposed \textit{T-Rex+DA} selector reliably controls the FDR in these heavy-tailed settings.

\section{Breast Cancer Survival Analysis}
\label{sec: FDR-Controlled Survival Analysis}
To illustrate the practical usefulness of the proposed \textit{T-Rex+DA} selector in large-scale high-dimensional settings, we demonstrate the performance of the proposed method via a breast cancer survival analysis using gene expression and survival time data from the open source resource The Cancer Genome Atlas (TCGA)~\cite{tomczak2015review,colaprico2016tcgabiolinks}. In order to detect the genes that are truly associated with the survival time of breast cancer patients, we conduct an FDR-controlled breast cancer survival analysis.

\subsection{TCGA Breast Cancer Data}
\label{subsec: TCGA Breast Cancer Data}
The gene expression levels are derived from the RNA-sequencing (RNA-seq) count data. The raw RNA-seq count data matrix $\X \in \mathbb{R}^{n \times p}$ contains $n = 1{,}095$ samples (i.e., breast cancer patients) and $p = 19{,}962$ protein coding genes. After two standard preprocessing steps, which are removing all genes with extremely low expression levels (i.e., where the sum of the RNA-seq counts is less than $10$) and performing a standard variance stabilizing transformation on the count data using the DESeq2 software~\cite{love2014moderated}, $p = 19{,}405$ candidate genes are left. The response vector $\y \in\mathbb{R}^{n}$ contains the log-transformed survival times of the patients. After removing missing and uninformative entries (i.e., entries with a survival time of zero days) from $\y$, $n = 1{,}072$ samples are left. During the study, the event (i.e., death) occurred for only $149$ patients, while $923$ patients were either still alive after the end of the study or dropped out of the study. That is, the survival times of $923$ patients are right censored. This is dealt with by treating these entries in $\y$ as missing data and imputing them using the well-known Buckley-James estimator~\cite{buckley1979linear}.
%
\begin{figure}[t]
  \centering
  \hspace*{-0.85em}
  \subfloat[\textit{T-Rex} methods]{
  		\scalebox{1}{
  			\includegraphics[width=0.476\linewidth]{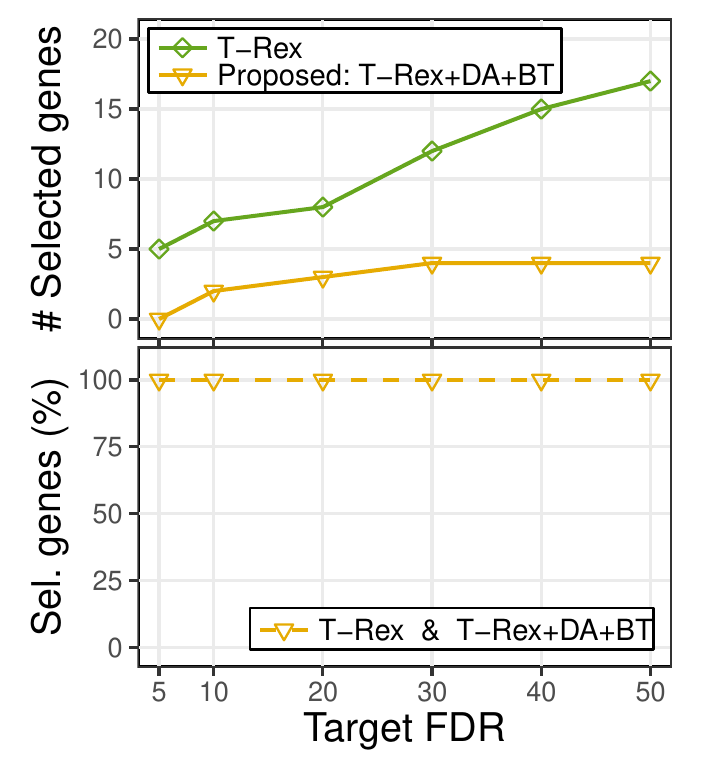}
  		}
   		\label{fig: TRex_survivalAnalysis}
   }
	\hspace*{-1.4em}
  \subfloat[Cox \textit{Lasso} and \textit{EN}]{
  		\scalebox{1}{
  			\includegraphics[width=0.506\linewidth]{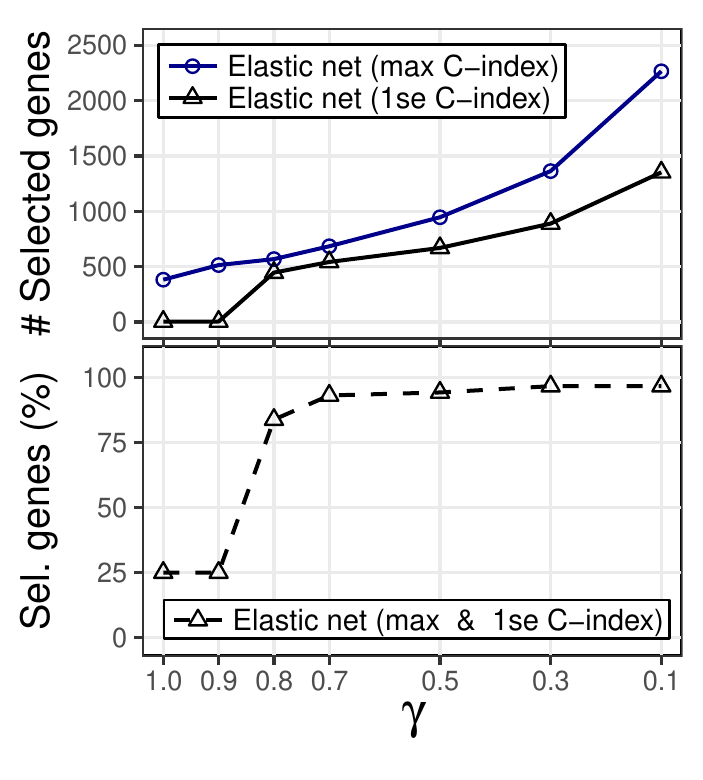}
  		}
   		\label{fig: enCox_survivalAnalysis}
   }
  \caption{Number of selected genes in the TCGA breast cancer survival analysis study.}
  \label{fig: numSelectedGenes_survivalAnalysis}
\end{figure}
%

\subsection{Methods and Results}
\label{subsec: Breast Cancer Study: Methods and Results}
As benchmark methods, we consider the Cox proportional hazards \textit{Lasso} and \textit{elastic net}~\cite{simon2011regularization}, which are specifically designed for censored survival data, and the ordinary \textit{T-Rex} selector~\cite{machkour2021terminating}. The elastic net Cox model requires the tuning of two parameters, i.e., a sparsity parameter $\lambda$ and a mixture parameter $\gamma \in [0, 1]$ that balances a convex combination of the $\ell_{1}$- and $\ell_{2}$-norm regularization terms. Here, $\gamma = 1$ sets the $\ell_{2}$ regularization term to zero and yields the Lasso solution. As suggested in~\cite{simon2011regularization}, we evaluate a range of values for $\gamma$ and, for each fixed $\gamma$, we perform $10$-fold cross-validation to choose $\lambda$. We consider, as suggested by the authors, the $\lambda$-value that achieves the maximum C-index and the $\lambda$-value that deviates by one standard error ($1$se criterion) from the maximum C-index to obtain a sparser solution. Due to the high computational complexity of the \textit{model-X} knockoff method (see, e.g., Figure~1 in~\cite{machkour2021terminating}), it is practically infeasible in this large-scale high-dimensional setting. Therefore, we cannot consider it in this survival analysis.

Figure~\ref{fig: numSelectedGenes_survivalAnalysis} shows the number of selected genes for different target FDR levels (in \%) and different values of $\gamma$. First, we observe that the ordinary \textit{T-Rex} method selected more genes than the proposed dependency-aware \textit{T-Rex} selector. In accordance with Corollary~\ref{Corollary: T-Rex+DA solution is subset of T-Rex solution}, all genes that were selected by the proposed method were also selected by the more liberal ordinary \textit{T-Rex} selector. In contrast, the regularized Cox methods did not provide consistent results for many values of $\gamma$ because many genes that were selected by the more conservative $1$se criterion do not appear in the selected set of the more liberal maximum C-index criterion. Moreover, it seems that many choices of $\gamma$ lead to a very high number of selected genes, which raises some suspicion with respect to reproducibility because only $149$ non-censored data points are usually not sufficient to reliably detect thousands of genes. By contrast, all three genes that were selected by the proposed method at a target FDR level of $20$\% (i.e., `ITM2A', `SCGB2A1', `RYR2') have been previously identified to be related to breast cancer~\cite{zhou2019integral,lacroix2006significance,xu2021bioinformatic}.

\section{Conclusion}
\label{sec: Conclusion}
The dependency-aware \textit{T-Rex} selector has been proposed. In contrast to existing methods, it reliably controls the FDR in the presence of groups of highly correlated variables in the data. A real world TCGA breast cancer survival analysis showed that the proposed method selects genes that have been previously identified to be related to breast cancer. Thus, the \textit{T-Rex+DA} selector is a promising tool for making reproducible discoveries in biomedical applications. Moreover, the derived group design principle allows to easily adapt the method to various application-specific dependency-structures, which opens the door to other fields that require large-scale high-dimensional variable selection with FDR-control guarantees. In fact, the group design principle was already successfully applied as a guiding principle in adapting the proposed \textit{T-Rex+DA} selector for FDR-controlled sparse financial index tracking~\cite{machkour2024TRexIndexTracking}.

\appendices

\section{Proofs and Technical Lemmas}
\label{appendix: Proofs and Technical Lemmas}

\subsection{Proof of Theorem 1}
\label{subsec: Proof of Theorem 1}
\begin{proof}
Note that for any $\hat{\beta}_{j, k} \neq 0$, the indicator function in~\eqref{eq: relative occurrences} can be written as follows:
\begin{equation}
\mathbbm{1}_{k}(j, T, L) = \big| \sign \big( \hat{\beta}_{j, k} \big) \big|, \quad j = 1, \ldots, p.
\label{eq: sign indicator function}
\end{equation}
Thus, we can rewrite the left-hand side of the inequality in Theorem~\ref{Theorem: absolute difference relative occurrences} as follows:
\begingroup
\allowdisplaybreaks
\begin{align}
&\dfrac{\big| \Phi_{T, L}(j) - \Phi_{T, L}(j^{\prime}) \big|}{\| \y \|_{2}}
\\
&\quad = \dfrac{1}{\| \y \|_{2} \cdot K} \Bigg| \sum\limits_{k = 1}^{K} \big( \mathbbm{1}_{k}(j, T, L) - \mathbbm{1}_{k}(j^{\prime}, T, L) \big) \Bigg|
\\
&\quad \leq \dfrac{1}{\| \y \|_{2} \cdot K} \sum\limits_{k = 1}^{K} \Big| \mathbbm{1}_{k}(j, T, L) - \mathbbm{1}_{k}(j^{\prime}, T, L) \Big|
\\
&\quad = \dfrac{1}{\| \y \|_{2} \cdot K} \sum\limits_{k = 1}^{K} \bigg| \big| \sign \big( \hat{\beta}_{j, k} \big) \big| - \big| \sign \big( \hat{\beta}_{j^{\prime}, k} \big) \big| \bigg|
\\
&\quad \leq \sqrt{2(1 - \rho_{j, j^{\prime}})} \cdot \dfrac{1}{K} \sum\limits_{k = 1}^{K} \dfrac{1}{\lambda_{k}(T, L)} \cdot \dfrac{\| \boldsymbol{\hat{r}}_{k} \|_{2}}{\| \y \|_{2}}
\\
&\quad \leq \sqrt{2 (1 - \rho_{j, j^{\prime}})} \cdot \widebar{\Lambda},
\label{eq: rewriting LHS of theorem - absolute difference relative occurrences}
\end{align}
\endgroup
where $\widebar{\Lambda} = \frac{1}{K} \sum_{k = 1}^{K} \frac{1}{\lambda_{k}(T, L)}$. The equation in the second line follows from the definition of the relative occurrences in~\eqref{eq: relative occurrences}, the inequality in the third line is a consequence of the triangle inequality, the equation in the fourth line follows from~\eqref{eq: sign indicator function}, and the inequality in the fifth line follows from Lemma~\ref{lemma: absolute sign difference}, which is deferred to Appendix~\ref{sec: Preliminaries for Theorem 1} in the supplementary materials. The inequality in the last line holds since $\hatbbeta_{k}$ is by definition the minimizer of~\eqref{eq: Lasso Lagrangian} and, therefore,
\begin{equation}
\mathcal{L} \big( \bbeta_{k} = \hatbbeta_{k}, \lambda_{k}(T, L) \big) \leq \mathcal{L} \big( \bbeta_{k} = \boldsymbol{0}, \lambda_{k}(T, L) \big)
\label{eq: put minimizer into Lagrangian}
\end{equation}
and, equivalently,
\begin{equation}
\dfrac{1}{2} \| \boldsymbol{\hat{r}}_{k} \|_{2}^{2} + \lambda_{k}(T, L) \| \hatbbeta_{k} \|_{1} \leq \dfrac{1}{2} \| \y \|_{2}^{2},
\label{eq: equivalent to eq: put minimizer into Lagrangian}
\end{equation}
which yields $\| \boldsymbol{\hat{r}}_{k} \|_{2} \leq \| \y \|_{2}$. Finally, we obtain
\begin{align}
\big| \Phi_{T, L}(j) - \Phi_{T, L}(j^{\prime}) \big| \leq \widebar{\Lambda} \| \y \|_{2} \cdot \sqrt{2 (1 - \rho_{j, j^{\prime}})}. \quad \qedhere
\label{eq: proof - Theorem - absolute difference relative occurrences}
\end{align}
\label{proof: theorem: absolute difference relative occurrences}
\end{proof}

\subsection{Technical Lemmas}
\label{subsec: Technical Lemmas}
\begin{lemma}
Define $\mathcal{V} \coloneq \lbrace \Phi_{T, L}^{\DA}(j, \rho_{\thr}(u_{\cut})) \geq 0.5, \, j = 1, \ldots, p \rbrace \backslash \lbrace 1 \rbrace$. Let 
\begin{align}
\mathcal{F}_{v} 
\coloneq 
\sigma \big( \big\lbrace V_{T, L}(u, &\rho_{\thr}(u_{\cut}))_{u \geq v} \big\rbrace,
\\
&\big\lbrace \widehat{V}_{T, L}^{\prime}(u, \rho_{\thr}(u_{\cut}))_{u \geq v} \big\rbrace \big)
\label{eq: backward-filtration}
\end{align}
be a backward-filtration with respect to $v$. Then, for all triples $(T, L, \rho_{\thr}(u_{\cut})) \in \lbrace 1, \ldots, L \rbrace \times \mathbb{N}_{+} \times [ 0, 1 ]$, $\lbrace H_{T, L}(v, \rho_{\thr}(u_{\cut})) \rbrace_{v \in \mathcal{V}}$ is a backward-running super-martingale with respect to $\mathcal{F}_{v}$. That is,
\begin{align}
\mathbb{E} \big[ H_{T, L}(v - \epsilon_{T, L}^{*}(v, \rho_{\thr}(&u_{\cut}))) \mid \mathcal{F}_{v} \big]
\\
&\geq H_{T, L}(v, \rho_{\thr}(u_{\cut})),
\label{eq: T-Rex+DA super-martingale}
\end{align}
where $v \in [ 0.5, 1 )$ and
\begin{align}
\epsilon_{T, L}^{*}(v, &\rho_{\thr}(u_{\cut})) \coloneq
\\
&\inf \big\lbrace \epsilon \in (0, v) : R_{T, L}(v - \epsilon, \rho_{\thr}(u_{\cut}))
\\
&\qquad\qquad\qquad\qquad - R_{T, L}(v, \rho_{\thr}(u_{\cut})) = 1\big\rbrace
\label{eq: martingale epsilon-step}
\end{align}
with the convention that $\epsilon_{T, L}^{*}(v, \rho_{\thr}(u_{\cut})) = 0$ if the infimum does not exist.
\label{Lemma: T-Rex+DA super-martingale}
\end{lemma}
\begin{proof}
The proof of Lemma~\ref{Lemma: T-Rex+DA super-martingale} follows along the lines of the proof of Lemma~5 in~\cite{machkour2021terminating} and by replacing $\Phi_{t, L}(j)$, $V_{T, L}(v)$, and $R_{T, L}(v)$ by their dependency-aware extensions $\Phi_{t, L}^{\DA}(j, \rho_{\thr}(u_{\cut}))$, $V_{T, L}(v, \rho_{\thr}(u_{\cut}))$, and $R_{T, L}(v, \rho_{\thr}(u_{\cut}))$, respectively.
\label{proof: Lemma - T-Rex+DA super-martingale}
\end{proof}
\begin{lemma}
Let $V_{T, L}^{+}(v, \rho_{\thr}(u_{\cut}))$ be as in~\eqref{eq: upper bound on V_T_L(0.5, rho_thr)}. For all triples $(T, L, v) \in \lbrace 1, \ldots, L \rbrace \times \mathbb{N}_{+} \times [0.5, 1)$, $V_{T, L}^{+}(v, \rho_{\thr}(u_{\cut}))$ is monotonically increasing in $\rho_{\thr}(u_{\cut})$, i.e., for any $u_{\cut} \in \lbrace 1, \ldots, p - 1 \rbrace$, it holds that
\begin{equation}
V_{T, L}^{+}(v, \rho_{\thr}(u_{\cut} + 1)) \geq V_{T, L}^{+}(v, \rho_{\thr}(u_{\cut})).
\label{eq: V_T,L monotonically increasing in rho}
\end{equation}
\label{Lemma: V_T,L monotonically increasing in rho}
\end{lemma}
\begin{proof}
Using the definition of $V_{T, L}^{+}(v, \rho_{\thr}(u_{\cut}))$ in~\eqref{eq: upper bound on V_T_L(0.5, rho_thr)}, we obtain
\begingroup
\allowdisplaybreaks
\begin{align}
V_{T, L}^{+}(&v, \rho_{\thr}(u_{\cut} + 1))
\\[0.5em]
\break
& = \big\lbrace \text{null } j : \Psi_{T, L}^{+} (j, \rho_{\thr}(u_{\cut} + 1)) \cdot \Phi_{T, L} (j) > v \big\rbrace
\\[0.5em]
& \geq \big\lbrace \text{null } j : \Psi_{T, L}^{+} (j, \rho_{\thr}(u_{\cut})) \cdot \Phi_{T, L} (j) > v \big\rbrace
\\[0.5em]
& = V_{T, L}^{+}(v, \rho_{\thr}(u_{\cut})).
\label{eq: proof - Lemma - V_T,L monotonically increasing in rho - 1}
\end{align}
\endgroup
The inequality in the third line follows from
\begin{equation}
\Psi_{T, L}^{+}(j, \rho_{\thr}(u_{\cut} + 1)) \geq \Psi_{T, L}^{+}(j, \rho_{\thr}(u_{\cut})),
\label{eq: Psi_T_L(j, rho_thr(u_c)) monotonically increasing in tho_thr(u_c)}
\end{equation}
which is a consequence of the following two cases:
\begin{align}
\text{(i) } &\Gr(j, \rho_{\thr}(u_{\cut})) = \varnothing:
\\[0.5em]
& \quad \Psi_{T, L}^{+}(j, \rho_{\thr}(u_{\cut} + 1)) = 1 = \Psi_{T, L}^{+}(j, \rho_{\thr}(u_{\cut})),
\\[1em]
\text{(ii) } &\Gr(j, \rho_{\thr}(u_{\cut})) \neq \varnothing:
\\[0.5em]
&\quad \Psi_{T, L}^{+} (j, \rho_{\thr}(u_{\cut} + 1))
\\[0.5em]
&\quad = \bigg[ 2 - \min_{j^{\prime} \in \Gr(j, \rho_{\thr}(u_{\cut} + 1))} \Big\lbrace \big| \Phi_{t, L}(j) - \Phi_{t, L}(j^{\prime}) \big| \Big\rbrace \bigg]^{-1}
\\
&\quad \geq \bigg[ 2 - \min_{j^{\prime} \in \Gr(j, \rho_{\thr}(u_{\cut}))} \Big\lbrace \big| \Phi_{t, L}(j) - \Phi_{t, L}(j^{\prime}) \big| \Big\rbrace \bigg]^{-1}
\\[0.5em]
&\quad = \Psi_{T, L}^{+} (j, \rho_{\thr}(u_{\cut})).
\label{eq: proof - Lemma - V_T,L monotonically increasing in rho - 2}
\end{align}
In~(ii), the inequality in the third line follows from the fact that, for any $u_{\cut} \in \lbrace 1, \ldots, p - 1 \rbrace$, it holds that
\begin{equation}
\Gr(j, \rho_{\thr}(u_{\cut} + 1)) \subseteq \Gr(j, \rho_{\thr}(u_{\cut})), \quad j = 1, \ldots, p,
\label{eq: proof - Lemma - V_T,L monotonically increasing in rho - 3}
\end{equation}
and, therefore,
\begin{align}
&\min_{j^{\prime} \in \Gr(j, \rho_{\thr}(u_{\cut} + 1))} \Big\lbrace \big| \Phi_{t, L}(j) - \Phi_{t, L}(j^{\prime}) \big| \Big\rbrace
\\[0.5em]
& \qquad\qquad\quad \geq \min_{j^{\prime} \in \Gr(j, \rho_{\thr}(u_{\cut}))} \Big\lbrace \big| \Phi_{t, L}(j) - \Phi_{t, L}(j^{\prime}) \big| \Big\rbrace. \qedhere
\label{eq: proof - Lemma - V_T,L monotonically increasing in rho - 4}
\end{align}
\label{proof: Lemma - V_T,L monotonically increasing in rho}
\end{proof}
\begin{lemma}
Let $\widehat{V}_{T, L}^{\prime \, +}(0.5, \rho_{\thr}(u_{\cut}))$ be as in~\eqref{eq: lower bound on V_hat_prime_T_L(0.5, rho_thr)}. For all tuples $(T, L) \in \lbrace 1, \ldots, L \rbrace \times \mathbb{N}_{+}$, $\widehat{V}_{T, L}^{\prime \, +}(0.5, \rho_{\thr}(u_{\cut}))$ is monotonically decreasing in $\rho_{\thr}(u_{\cut})$, i.e., for any $u_{\cut} \in \lbrace 1, \ldots, p - 1 \rbrace$, it holds that
\begin{equation}
\widehat{V}_{T, L}^{\prime \, +}(0.5, \rho_{\thr}(u_{\cut} + 1)) \leq \widehat{V}_{T, L}^{\prime \, +}(0.5, \rho_{\thr}(u_{\cut})).
\label{Lemma: V'_hat_T,L monotonically decreasing in rho}
\end{equation}
\end{lemma}
\begin{proof}
Using Equation~\eqref{eq: V_T_L_prime}, we obtain
\begingroup
\allowdisplaybreaks
\begin{align}
& \widehat{V}_{T, L}^{\prime \, +}(0.5, \rho_{\thr}(u_{\cut} + 1))
\\[0.5em]
& \quad = \sum\limits_{t = 1}^{T} \dfrac{p - \sum_{q = 1}^{p} \Psi_{t, L}^{+}(q, \rho_{\thr}(u_{\cut} + 1)) \cdot \Phi_{t, L}(q)}{L - (t - 1)}
\\[0.5em]
& \quad \leq \sum\limits_{t = 1}^{T} \dfrac{p - \sum_{q = 1}^{p} \Psi_{t, L}^{+}(q, \rho_{\thr}(u_{\cut})) \cdot \Phi_{t, L}(q)}{L - (t - 1)}
\\[0.5em]
& \quad = \widehat{V}_{T, L}^{\prime \, +}(0.5, \rho_{\thr}(u_{\cut})),
\label{eq: proof - Lemma - V'_hat_T,L monotonically decreasing in rho - 3}
\end{align}
\endgroup
where the inequality in the third line follows from
\begin{equation}
\Psi_{t, L}^{+}(q, \rho_{\thr}(u_{\cut} + 1)) \geq \Psi_{t, L}^{+}(q, \rho_{\thr}(u_{\cut})),
\label{eq: proof - Lemma - V'_hat_T,L monotonically decreasing in rho - 4}
\end{equation}
which was shown to hold within the proof of Lemma~\ref{Lemma: V_T,L monotonically increasing in rho}.
\label{proof: Lemma - V'_hat_T,L monotonically decreasing in rho}
\end{proof}
%

\bibliographystyle{IEEEtran}
\bibliography{bibliography}

\typeout{get arXiv to do 4 passes: Label(s) may have changed. Rerun}

\end{document}


\title{Supplement to\\``High-Dimensional False Discovery Rate Control\\for Dependent Variables''}

\author{
Jasin Machkour, 
Michael Muma, 
and
Daniel P. Palomar
\thanks{J. Machkour,~Student Member,~IEEE, and M. Muma, Senior Member,~IEEE, are with the Robust Data Science Group at Technische Universit\"at Darmstadt, Germany (e-mail: jasin.machkour@tu-darmstadt.de; michael.muma@tu-darmstadt.de). D. P. Palomar,~Fellow,~IEEE, is with the Convex Optimization Group, The Hong Kong University of Science and Technology, Hong Kong SAR, China (e-mail: palomar@ust.hk).}
\thanks{J. Machkour is supported by the LOEWE initiative (Hesse, Germany) within the emergenCITY center. M. Muma is supported by the ERC Starting Grant ScReeningData (Project Number: 101042407). D.P. Palomar is supported by the Hong Kong GRF 16206123 research grant.}
\thanks{Extensive computations on the Lichtenberg High-Performance Computer of the Technische Universität Darmstadt were conducted for this research.}
}

\maketitle

\appendices
\setcounter{section}{1}

\section{Preliminaries for Theorem 1}
\label{sec: Preliminaries for Theorem 1}
\begin{lemma}
Let $\rho_{j, j^{\prime}} \coloneqq \x_{j}^{\top} \x_{j^{\prime}}$, $j, j^{\prime} \in \lbrace 1, \ldots, p \rbrace$, be the sample correlation coefficient of the standardized variables $j$ and $j^{\prime}$, $\boldsymbol{\hat{r}}_{k} \coloneqq \y - \XWK_{k} \hatbbeta_{k}$, and $\sign \big( \hat{\beta}_{j, k} \big)$, $\sign \big( \hat{\beta}_{j^{\prime}, k} \big)$ be the signs of the $j$th and $j^{\prime}$th \textit{Lasso} coefficient estimates of the $k$th random experiment, respectively. Suppose that $\hat{\beta}_{j, k}, \hat{\beta}_{j^{\prime}, k} \neq 0$. Then, it holds that
\begin{equation}
\bigg| \big| \sign \big( \hat{\beta}_{j, k} \big) \big| - \big| \sign \big( \hat{\beta}_{j^{\prime}, k} \big) \big| \bigg| \leq \dfrac{\| \boldsymbol{\hat{r}}_{k} \|_{2} \sqrt{2\big( 1 - \rho_{j, j^{\prime}} \big)}}{\lambda_{k}(T, L)}.
\label{eq: absolute sign difference}
\end{equation}
\label{lemma: absolute sign difference}
\end{lemma}
\begin{proof}
The \textit{Lasso} optimization problem is solved by the coefficient vector $\bbeta_{k} = \hatbbeta_{k}$ that minimizes the function
\begin{equation}
\mathcal{L} \big( \bbeta_{k}, \lambda_{k}(T, L) \big) \coloneqq \dfrac{1}{2} \| \y - \XWK_{k} \bbeta_{k} \|_{2}^{2} + \lambda_{k}(T, L) \| \bbeta_{k} \|_{1}.
\label{eq: Lasso Lagrangian}
\end{equation}
Taking the first derivative of \eqref{eq: Lasso Lagrangian} and setting it equal to zero yields
\begin{equation}
\dfrac{\partial \mathcal{L} \big( \bbeta_{k}, \lambda_{k}(T, L) \big)}{\partial  \bbeta_{k}} \Biggr|_{\bbeta_{k} = \hatbbeta_{k}} 
\hspace{-1em}= -\XWK_{k}^{\top} \boldsymbol{\hat{r}}_{k} + \lambda_{k}(T, L) \dfrac{\partial \| \hatbbeta_{k} \|_{1}}{\partial \hatbbeta_{k}}
\overset{!}{=} 0,
\label{eq: first derivative Lasso Lagrangian}
\end{equation}
which is a system of equations whose $j$th and $j^{\prime}$th equation are given by
\begin{equation}
-\x_{j}^{\top} \boldsymbol{\hat{r}}_{k} + \lambda_{k}(T, L) \cdot \sign \big( \hat{\beta}_{j, k} \big) = 0,
\label{eq: jth equation of system of equations}
\end{equation}
\begin{equation}
-\x_{j^{\prime}}^{\top} \boldsymbol{\hat{r}}_{k} + \lambda_{k}(T, L) \cdot \sign \big( \hat{\beta}_{j^{\prime}, k} \big) = 0.
\label{eq: mth equation of system of equations}
\end{equation}
Subtracting Equation~\eqref{eq: mth equation of system of equations} from Equation~\eqref{eq: jth equation of system of equations} and rearranging the resulting equation yields
\begin{equation}
\sign \big( \hat{\beta}_{j, k} \big) - \sign \big( \hat{\beta}_{j^{\prime}, k} \big) = \dfrac{1}{\lambda_{k}(T, L)} \big( \x_{j} - \x_{j^{\prime}} \big)^{\top} \boldsymbol{\hat{r}}_{k}.
\label{eq: jth minus mth equation}
\end{equation}
Now, we can rewrite the left-hand side of~\eqref{eq: absolute sign difference} as follows:
\begin{align}
\bigg| \big| \sign \big( \hat{\beta}_{j, k} \big) \big| &- \big| \sign \big( \hat{\beta}_{j^{\prime}, k} \big) \big| \bigg|
\\
&\leq \Big| \sign \big( \hat{\beta}_{j, k} \big) - \sign \big( \hat{\beta}_{j^{\prime}, k} \big) \bigg|
\\
&= \dfrac{1}{\lambda_{k}(T, L)} \Big| \big( \x_{j} - \x_{j^{\prime}} \big)^{\top} \boldsymbol{\hat{r}}_{k} \Big|
\\
&\leq \dfrac{1}{\lambda_{k}(T, L)} \| \x_{j} - \x_{j^{\prime}} \|_{2} \cdot \| \boldsymbol{\hat{r}}_{k} \|_{2}
\\
&= \dfrac{\| \boldsymbol{\hat{r}}_{k} \|_{2} \sqrt{2\big( 1 - \rho_{j, j^{\prime}} \big)}}{\lambda_{k}(T, L)}.
\label{eq: rewriting LHS of lemma: 1 - absolute sign difference}
\end{align}
The inequality in the second line follows from the reverse triangle inequality, the equation in the third line follows from \eqref{eq: jth minus mth equation}, the inequality in the fourth line follows from the Cauchy-Schwartz inequality, and the equation in the last line is a consequence of
\begingroup
\allowdisplaybreaks
\begin{align}
\| \x_{j} - \x_{j^{\prime}} \|_{2}^{2} 
&= (\x_{j} - \x_{j^{\prime}})^{\top} (\x_{j} - \x_{j^{\prime}})
\\
&= \| \x_{j} \|_{2}^{2} + \| \x_{j^{\prime}} \|_{2}^{2} - 2 \x_{j}^{\top} \x_{j^{\prime}}
\\
&= 2(1 - \rho_{j, j^{\prime}}),
\label{eq: l2-norm of x_j - x_m}
\end{align}
\endgroup
where the last line follows from the fact that the variables are standardized.
\label{proof: lemma: 1 - absolute sign difference}
\end{proof}

\section{Additional Simulations}
\label{appendix: Additional Simulations}
In this appendix, we consider two additional heavy-tailed simulation settings to complement the numerical experiments presented in Section~\ref{sec: Numerical Experiments}. These settings differ from the simulation setting in Section~\ref{sec: Numerical Experiments} as follows:
\begin{enumerate}
\item Heavy-tailed predictor matrix $\X$: The $p$-dimensional samples of the predictor matrix (i.e., rows of $\X$) are sampled from a zero-mean multivariate heavy-tailed Student-$t$ distribution with covariance matrix $\bSigma$ (with $\rho = 0.8$) and $3$ degrees of freedom,
\item Heavy-tailed noise vector $\bepsilon$: The noise vector $\bepsilon$ is sampled from a heavy-tailed Student-$t$ distribution with $3$ degrees of freedom.
\end{enumerate}
The additional simulation results in Figures~\ref{fig: sweep_plots_3_tFDR_20_n_150_rho_07_tDistrData}, \ref{fig: sweep_plots_3_tFDR_20_n_150_rho_07_tDistrNoise}, \ref{fig: sweep_plots_tFDR_tDistrData}, and~\ref{fig: sweep_plots_tFDR_tDistrNoise} verify the theoretical results in Section~\ref{sec: Methodology and Main Theoretical Results} and show that only the proposed \textit{T-Rex+DA} selector reliably controls the FDR in these heavy-tailed settings.
%
\begin{figure*}[!htbp]
  \centering
  \subfloat[]{
  		\scalebox{1}{
  			\includegraphics[width=0.37\linewidth]{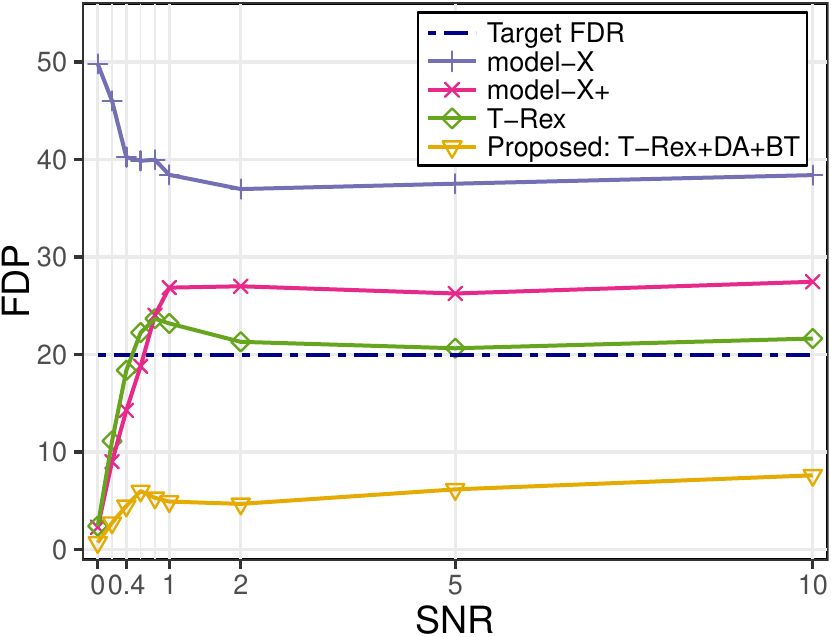}
  		}
   		\label{fig: FDP_vs_SNR_tFDR_20_n_150_rho_07_tDistrData}
   }
	\hspace*{2em}
  \subfloat[]{
  		\scalebox{1}{
  			\includegraphics[width=0.3765\linewidth]{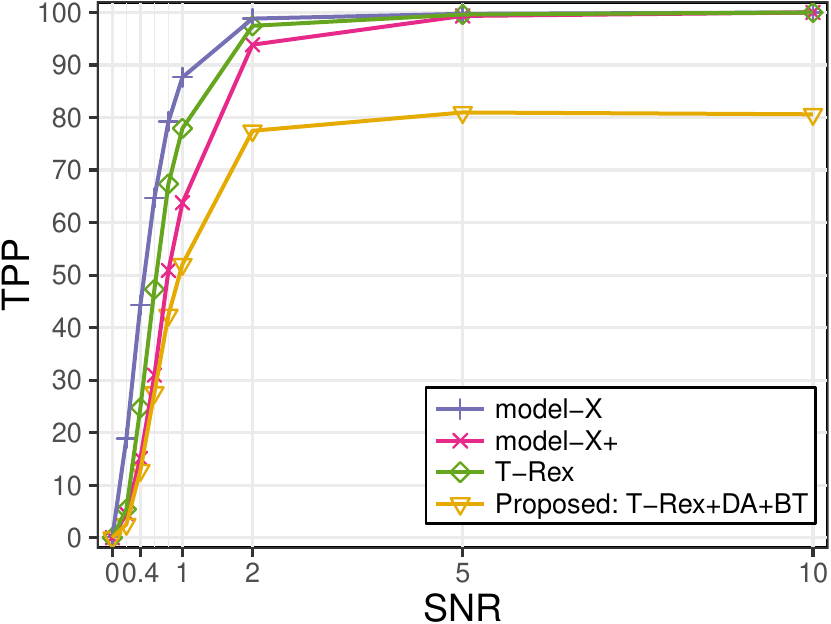}
  		}
   		\label{fig: TPP_vs_SNR_tFDR_20_n_150_rho_07_tDistrData}
   }
   \\
\vspace{-0.8em}
  \subfloat[]{
  		\scalebox{1}{
  			\includegraphics[width=0.37\linewidth]{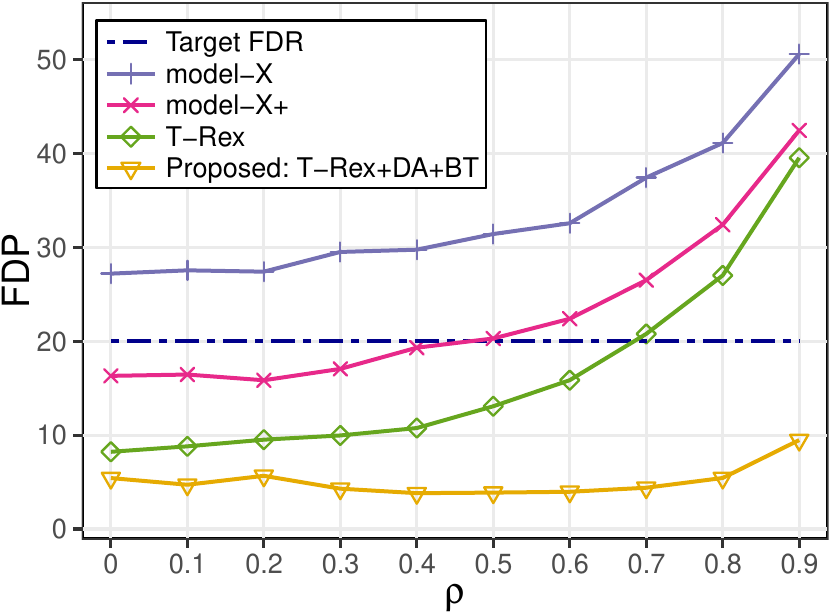}
  		}
   		\label{fig: FDP_vs_rho_tFDR_20_n_150_rho_07_tDistrData}
   }
	\hspace*{2em}
  \subfloat[]{
  		\scalebox{1}{
  			\includegraphics[width=0.3765\linewidth]{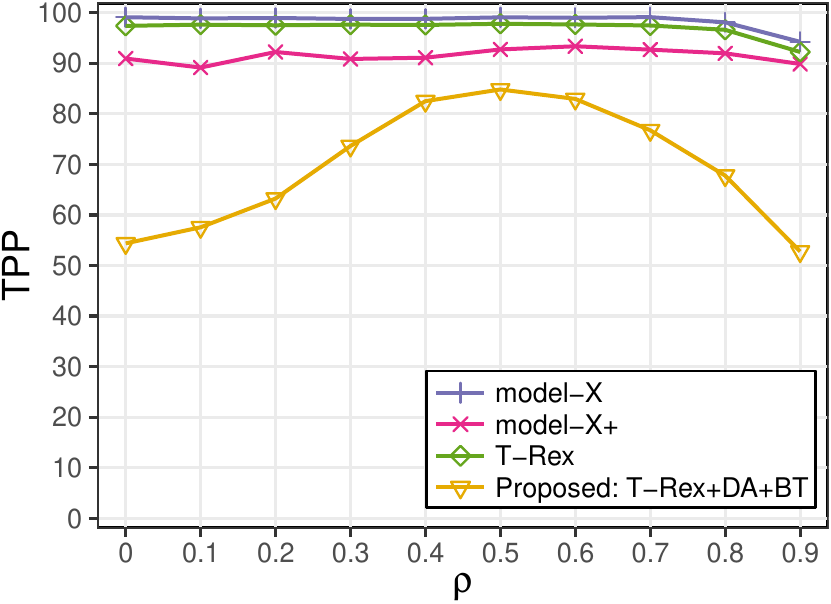}
  		}
   		\label{fig: TPP_vs_rho_tFDR_20_n_150_rho_07_tDistrData}
   }
   \\
\vspace{-0.8em}
  \subfloat[]{
  		\scalebox{1}{
  			\includegraphics[width=0.37\linewidth]{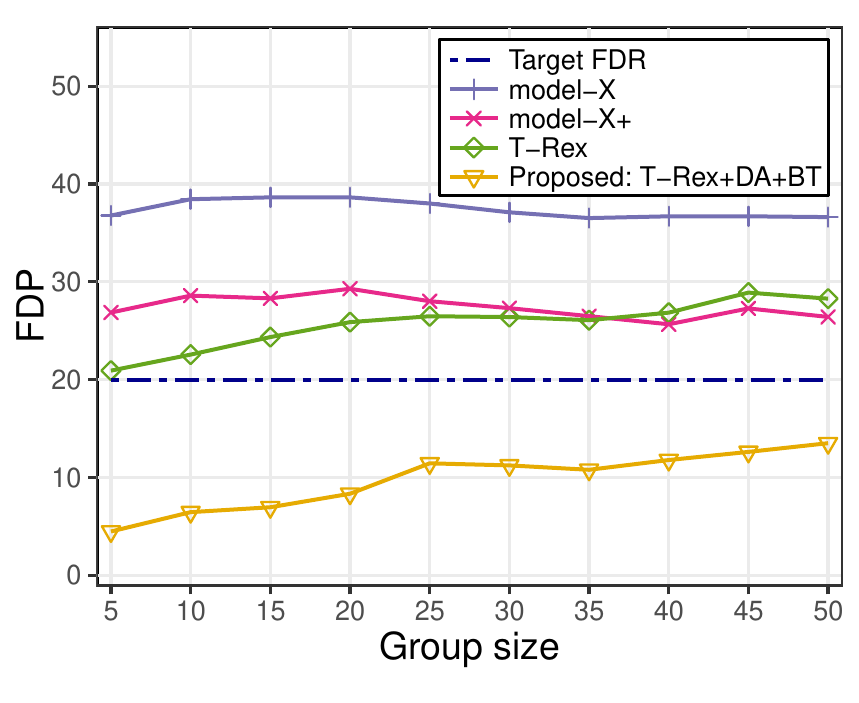}
  		}
   		\label{fig: FDP_vs_groupSize_tFDR_20_n_150_rho_07_tDistrData}
   }
	\hspace*{2em}
  \subfloat[]{
  		\scalebox{1}{
  			\includegraphics[width=0.3765\linewidth]{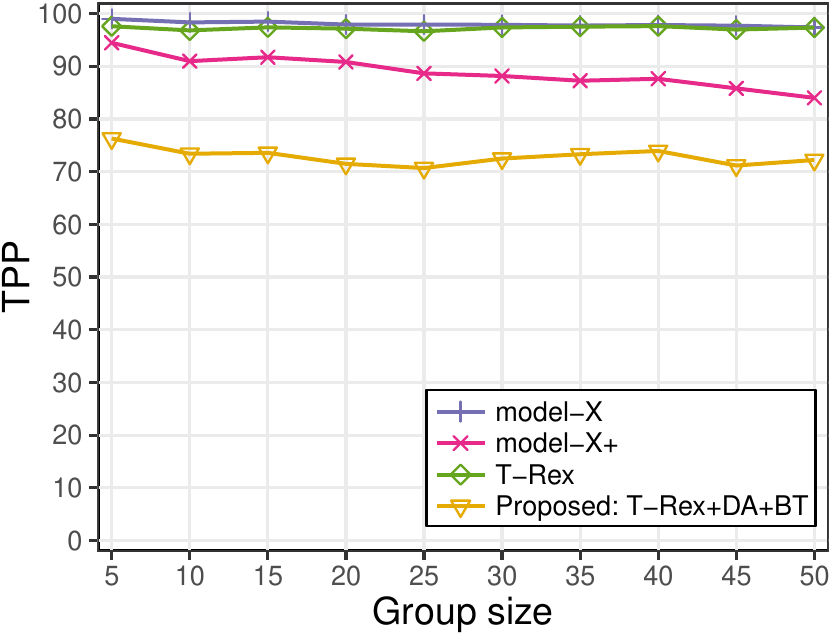}
  		}
   		\label{fig: TPP_vs_groupSize_tFDR_20_n_150_rho_07_tDistrData}
   }
   \\
\vspace{-0.8em}
  \subfloat[]{
  		\scalebox{1}{
  			\includegraphics[width=0.37\linewidth]{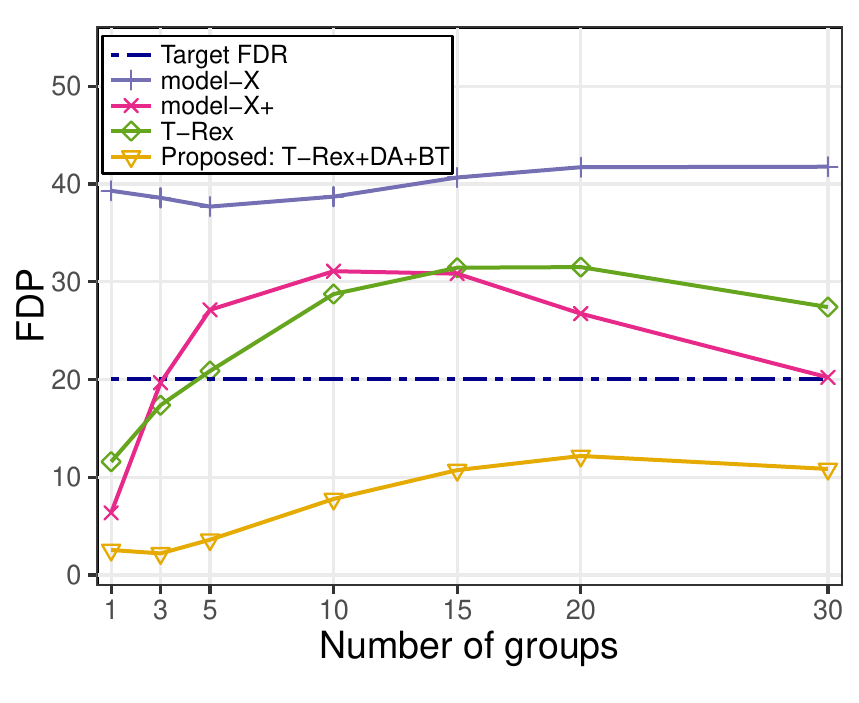}
  		}
   		\label{fig: FDP_vs_numGroups_tFDR_20_n_150_rho_07_tDistrData}
   }
	\hspace*{2em}
  \subfloat[]{
  		\scalebox{1}{
  			\includegraphics[width=0.3765\linewidth]{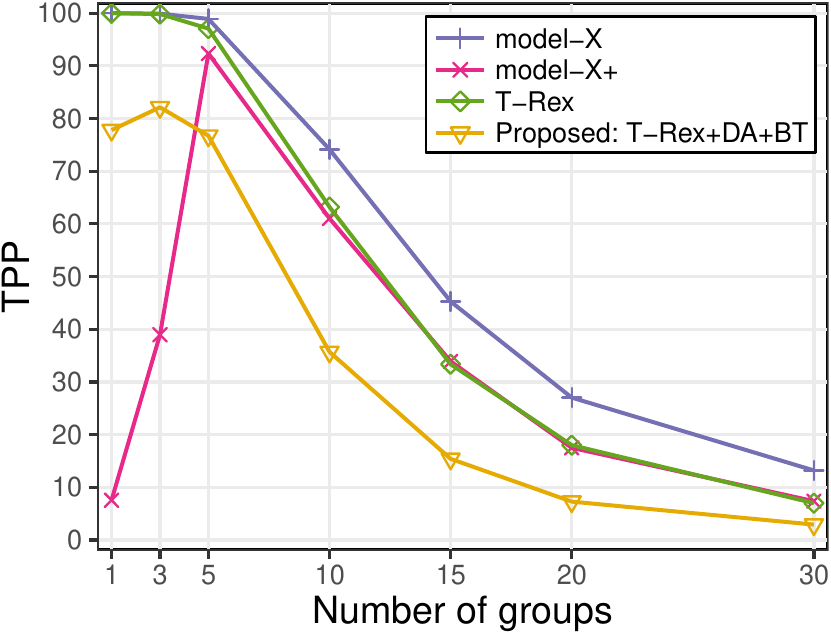}
  		}
   		\label{fig: TPP_vs_numGroups_tFDR_20_n_150_rho_07_tDistrData}
   }
  \caption{Heavy-tailed predictor matrix $\X$.
}
  \label{fig: sweep_plots_3_tFDR_20_n_150_rho_07_tDistrData}
\end{figure*}
%
\begin{figure*}[b]
  \centering
  \subfloat[]{
  		\scalebox{1}{
  			\includegraphics[width=0.37\linewidth]{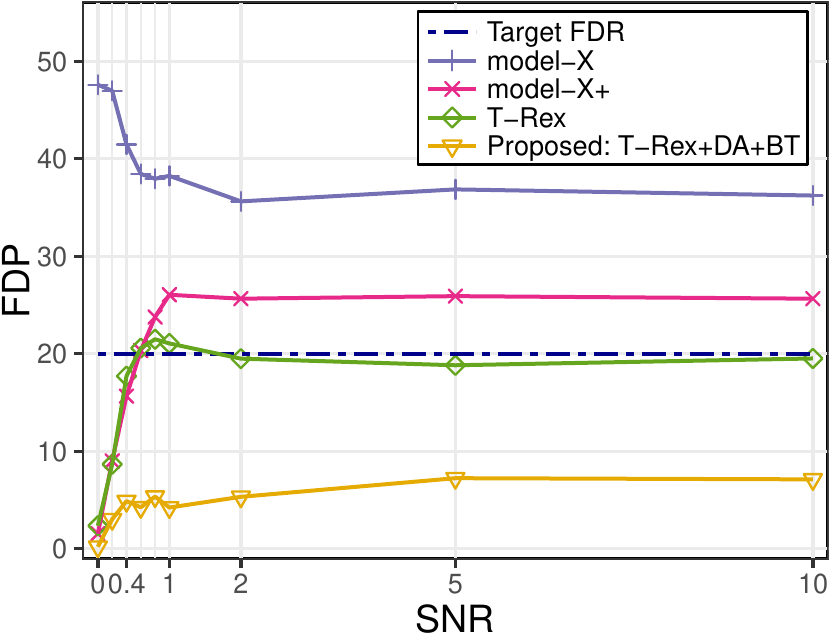}
  		}
   		\label{fig: FDP_vs_SNR_tFDR_20_n_150_rho_07_tDistrNoise}
   }
	\hspace*{2em}
  \subfloat[]{
  		\scalebox{1}{
  			\includegraphics[width=0.3765\linewidth]{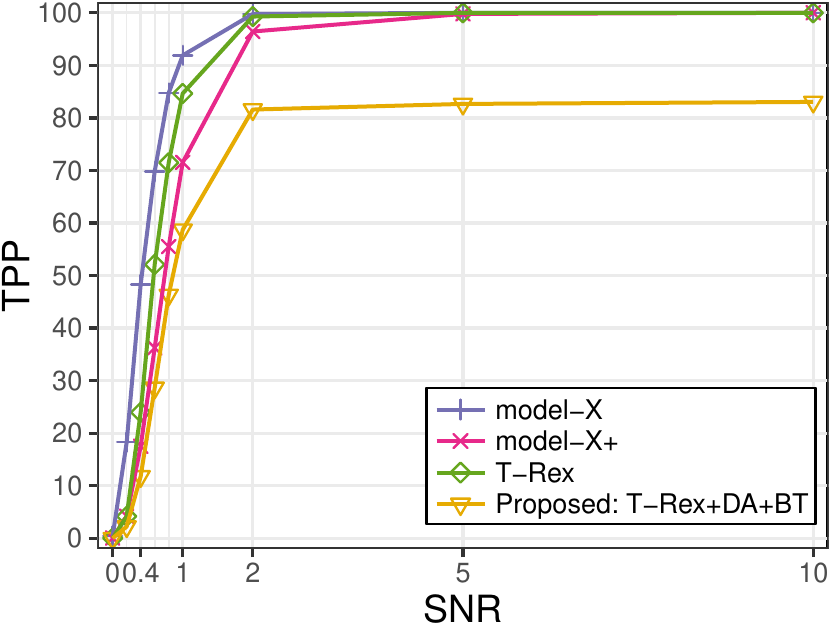}
  		}
   		\label{fig: TPP_vs_SNR_tFDR_20_n_150_rho_07_tDistrNoise}
   }
   \\
\vspace{-0.8em}
  \subfloat[]{
  		\scalebox{1}{
  			\includegraphics[width=0.37\linewidth]{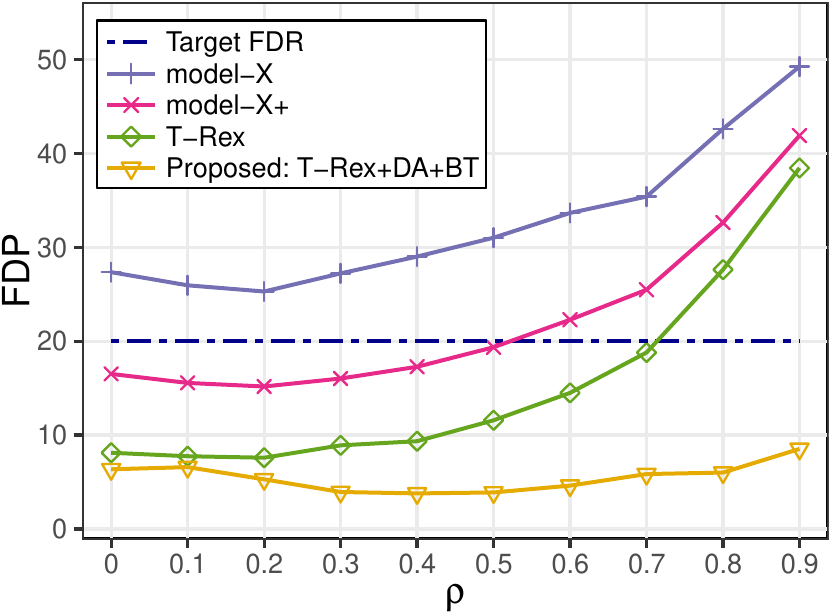}
  		}
   		\label{fig: FDP_vs_rho_tFDR_20_n_150_rho_07_tDistrNoise}
   }
	\hspace*{2em}
  \subfloat[]{
  		\scalebox{1}{
  			\includegraphics[width=0.3765\linewidth]{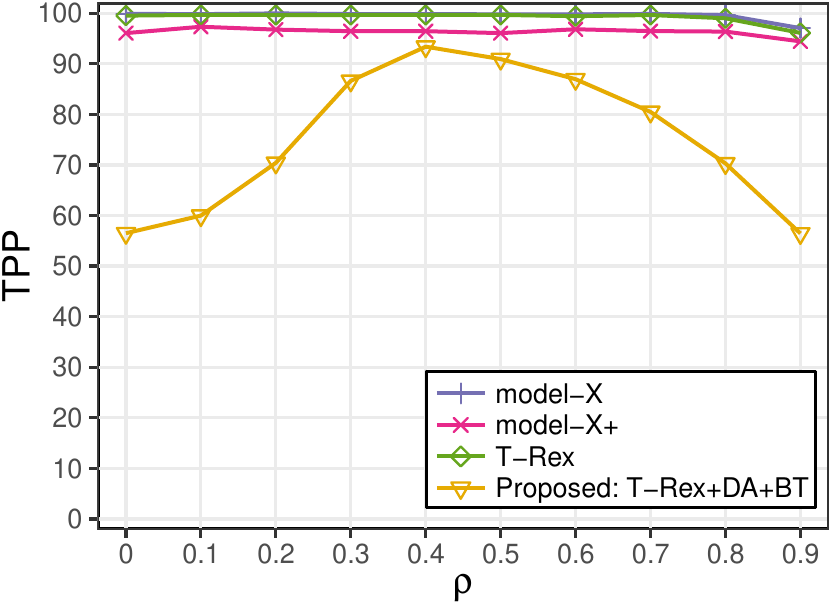}
  		}
   		\label{fig: TPP_vs_rho_tFDR_20_n_150_rho_07_tDistrNoise}
   }
   \\
\vspace{-0.8em}
  \subfloat[]{
  		\scalebox{1}{
  			\includegraphics[width=0.37\linewidth]{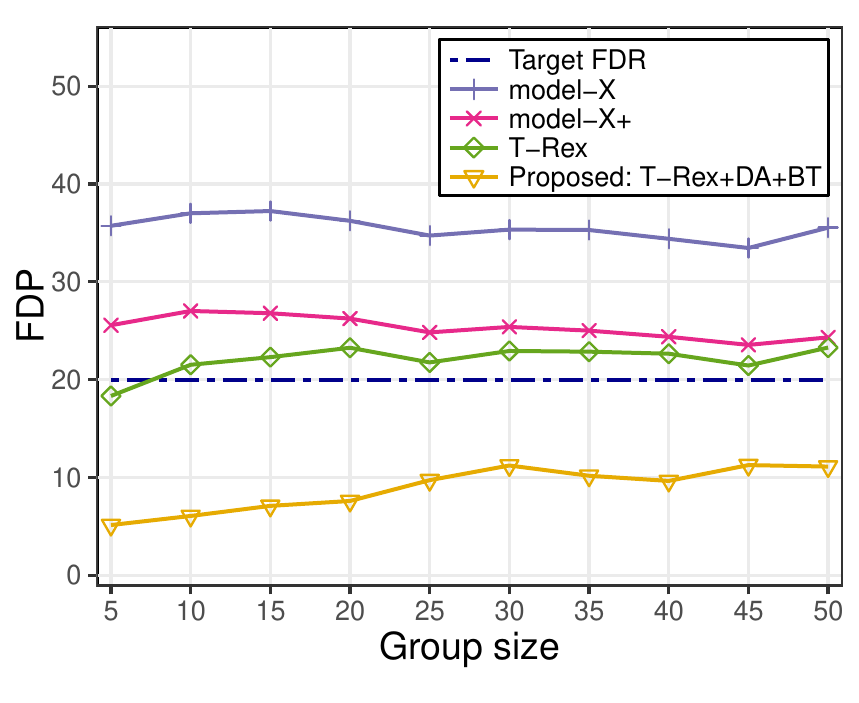}
  		}
   		\label{fig: FDP_vs_groupSize_tFDR_20_n_150_rho_07_tDistrNoise}
   }
	\hspace*{2em}
  \subfloat[]{
  		\scalebox{1}{
  			\includegraphics[width=0.3765\linewidth]{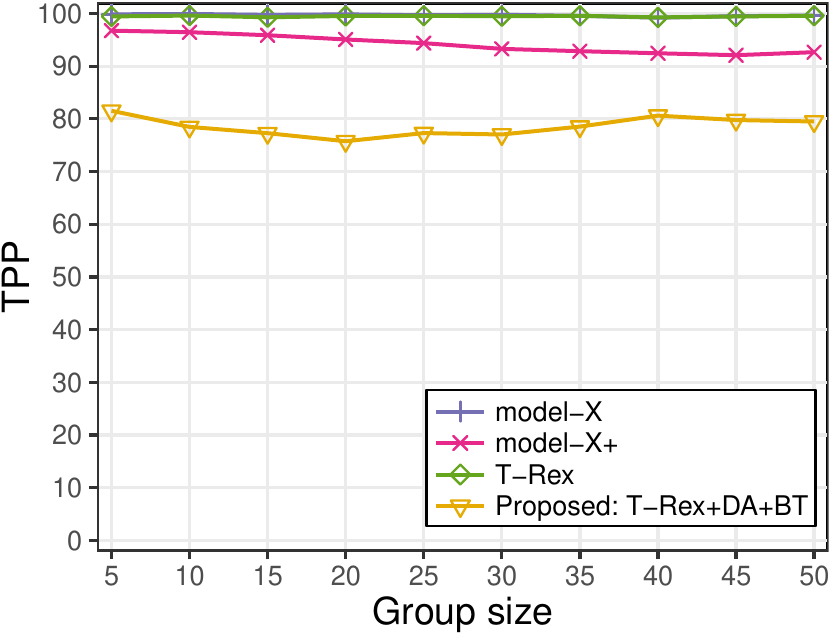}
  		}
   		\label{fig: TPP_vs_groupSize_tFDR_20_n_150_rho_07_tDistrNoise}
   }
   \\
\vspace{-0.8em}
  \subfloat[]{
  		\scalebox{1}{
  			\includegraphics[width=0.37\linewidth]{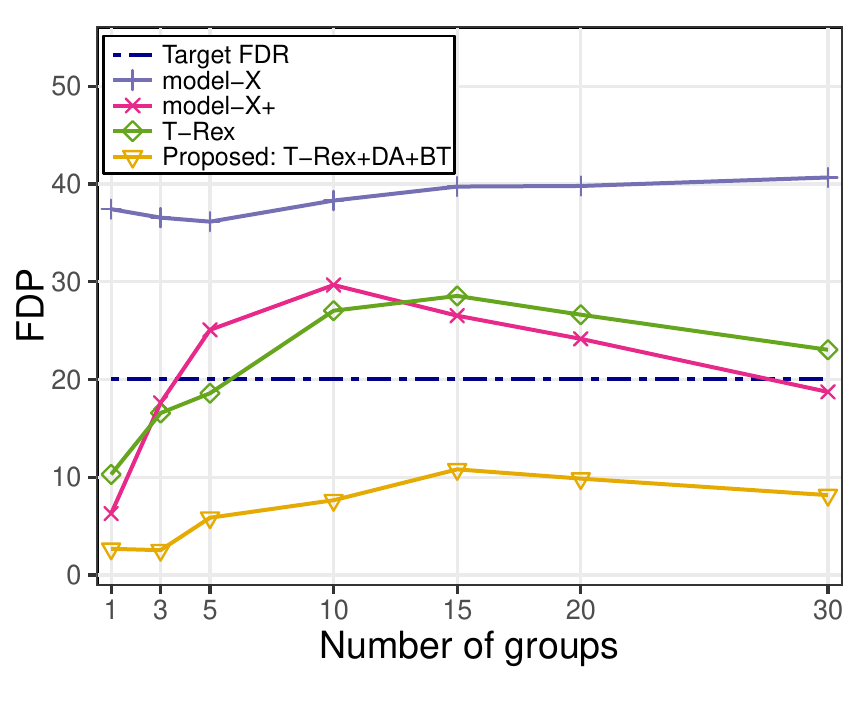}
  		}
   		\label{fig: FDP_vs_numGroups_tFDR_20_n_150_rho_07_tDistrNoise}
   }
	\hspace*{2em}
  \subfloat[]{
  		\scalebox{1}{
  			\includegraphics[width=0.3765\linewidth]{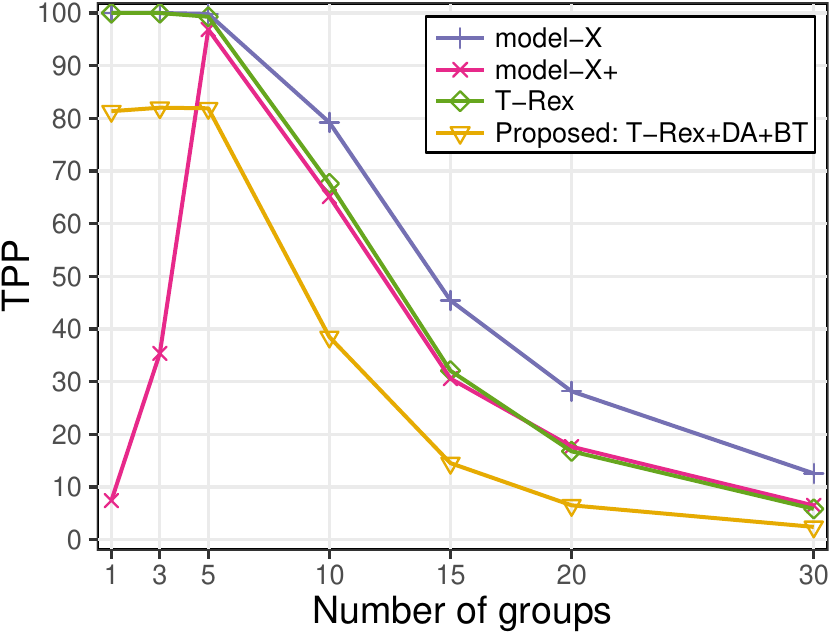}
  		}
   		\label{fig: TPP_vs_numGroups_tFDR_20_n_150_rho_07_tDistrNoise}
   }
  \caption{Heavy-tailed noise vector $\bepsilon$.
}
  \label{fig: sweep_plots_3_tFDR_20_n_150_rho_07_tDistrNoise}
\end{figure*}
%
\begin{figure*}[h]
  \centering
  \hspace*{-0.75em}
  \subfloat[]{
  		\scalebox{1}{
  			\includegraphics[width=0.37\linewidth]{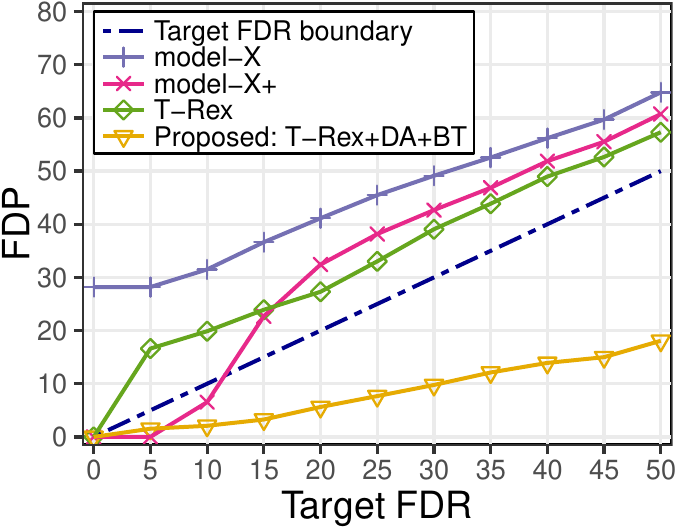}
  		}
   		\label{fig: FDP_vs_tFDR_n_150_p_500_rho_07_tDistrData}
   }
	\hspace*{-1.3em}
  \subfloat[]{
  		\scalebox{1}{
\includegraphics[width=0.37\linewidth]{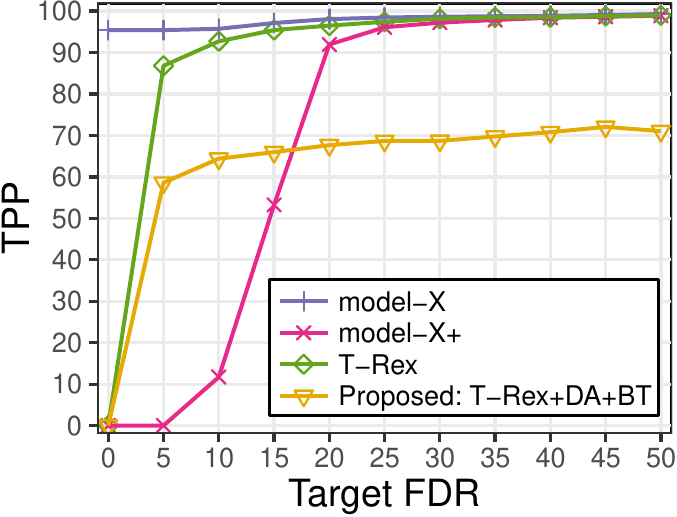}
  		}
   		\label{fig: TPP_vs_tFDR_n_150_p_500_rho_07_tDistrData}
   }
  \caption{Heavy-tailed predictor matrix $\X$.
  }
  \label{fig: sweep_plots_tFDR_tDistrData}
\end{figure*}
%
\begin{figure*}[h]
  \centering
  \hspace*{-0.75em}
  \subfloat[]{
  		\scalebox{1}{
  			\includegraphics[width=0.37\linewidth]{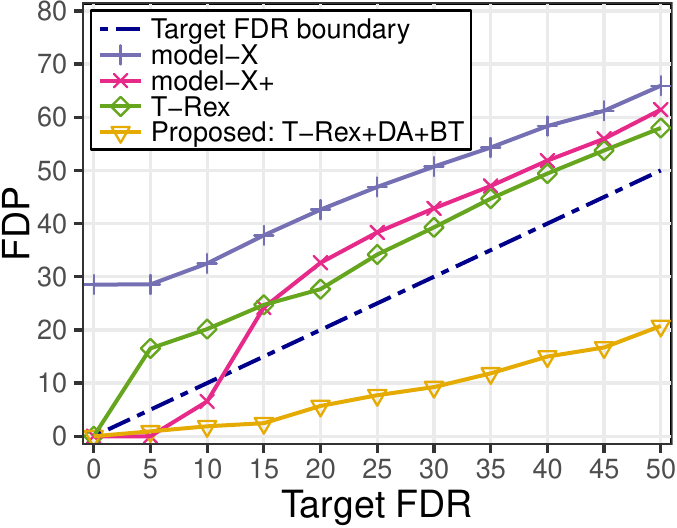}
  		}
   		\label{fig: FDP_vs_tFDR_n_150_p_500_rho_07_tDistrNoise}
   }
	\hspace*{-1.3em}
	\subfloat[]{
  		\scalebox{1}{
\includegraphics[width=0.37\linewidth]{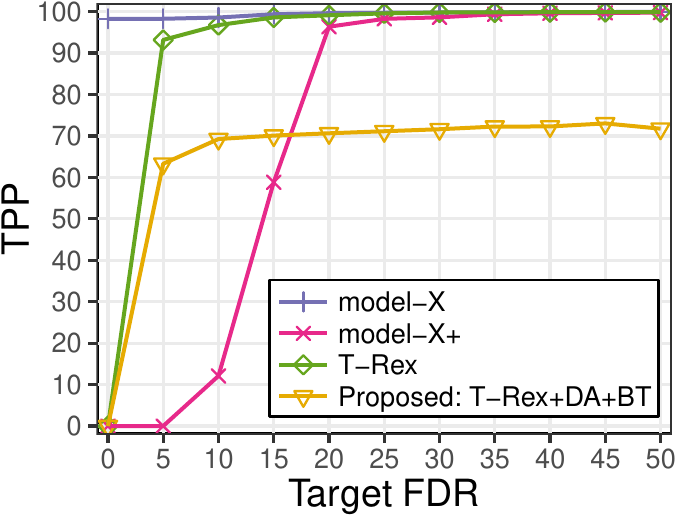}
  		}
   		\label{fig: TPP_vs_tFDR_n_150_p_500_rho_07_tDistrNoise}
   }
  \caption{Heavy-tailed noise vector $\bepsilon$.
  }
  \label{fig: sweep_plots_tFDR_tDistrNoise}
\end{figure*}
%
\pagebreak
\vfill

\typeout{get arXiv to do 4 passes: Label(s) may have changed. Rerun}